\RequirePackage{etoolbox}
\csdef{input@path}{%
 {sty/}%
 {img/}%
}%
\csgdef{bibdir}{bib/}%

\documentclass[ba]{imsart}
\pubyear{0000}
\volume{00}
\issue{0}
\doi{0000}
\firstpage{1}
\lastpage{1}

\usepackage{amsthm}
\usepackage[colorlinks,citecolor=blue,urlcolor=blue,filecolor=blue,backref=page]{hyperref}

\usepackage{graphicx, verbatim}
\usepackage{caption}
\usepackage{subcaption}
\usepackage{amsmath}
\usepackage{bbm}

\usepackage{natbib}
\usepackage{booktabs}
\usepackage{float}
\usepackage{amsmath}
\usepackage{bm}
\usepackage{amssymb}
\usepackage[yyyymmdd,hhmmss]{datetime}
\usepackage{dsfont}
\usepackage{algorithm}
\usepackage{algorithmic}

\usepackage{amsmath}
\usepackage{amsthm}
\usepackage{amssymb}

\newtheorem{proposition}{Proposition}

\newtheorem{corollary}{Corollary}

\newcommand{\OO}{\mathcal{O}}

\DeclareMathOperator{\PG}{PG}
\DeclareMathOperator{\Var}{Var}
\DeclareMathOperator{\diag}{diag}

\begin{document}

\begin{frontmatter}
    \title{Distributed Computation for Marginal Likelihood based Model Choice \thanksref{T1}}
    \runtitle{Distributed Computation for Marginal Likelihood based Model Choice}

    \begin{aug}
    \author{\fnms{Alexander Buchholz}\thanksref{addr4,addr1}\ead[label=e1]{ab2603@cantab.ac.uk}},
    \author{\fnms{Daniel Ahfock}\thanksref{addr4,addr3}}%
    \and
    \author{\fnms{Sylvia Richardson}\thanksref{addr2}
    \ead[label=e3]{third@somewhere.com}
    \ead[label=u1,url]{http://www.foo.com}}
    
    \runauthor{Buchholz, Ahfock and Richardson}
    
    \address[addr4]{Equal contributions
    }

    \address[addr1]{Amazon Berlin, Germany, work done while at the University of Cambridge,
        \printead{e1} %
    }
    
    \address[addr2]{MRC Biostatistics Unit, University of Cambridge, UK
    }
    
    \address[addr3]{School of Mathematics and Physics, The University of Queensland, AUS
    }

    \end{aug}
    
    \begin{abstract}
        We propose a general method for distributed Bayesian model choice, using the marginal likelihood, 
        where a data set is split in non-overlapping subsets. These subsets are only accessed locally by individual workers and no data is shared between the workers. 
        We approximate the model evidence for the full data set through Monte Carlo sampling from the posterior 
        on every subset generating a model evidence per subset. The results are 
        combined using a novel approach which corrects for the splitting using summary statistics 
        of the generated samples. 
        Our divide-and-conquer approach enables Bayesian model choice in the large data setting, 
        exploiting all available information but limiting communication between workers. We derive theoretical error bounds that quantify the resulting trade-off between computational gain and loss in precision. 
        The embarrassingly parallel nature yields important speed-ups when used on massive data sets as illustrated by our real world experiments. 
        In addition, we show how 
        the suggested approach can be extended to model choice within a reversible jump setting that explores multiple feature combinations within one run. 

    \end{abstract}
    
    \begin{keyword}[class=MSC]
    \kwd[Primary ]{62F15}
    \kwd[; secondary ]{68W15}
    \end{keyword}
    
    \begin{keyword}
    \kwd{Marginal likelihood}
    \kwd{Distributed computation}
    \end{keyword}
    
\end{frontmatter}

\section{Introduction}
The marginal likelihood, also known as the posterior normalising constant or the Bayesian model evidence, is a quantity that is notoriously difficult to calculate, 
but crucial for model selection in a Bayesian setting \citep{kass1995bayes, robert2007bayesian}.  Evaluation of the marginal likelihood is particularly computationally challenging with massive data sets that are too large to fit in the memory of a single computational node. 
Distributed computing is attractive in such a situation, as the data set can be split into smaller manageable subsets which can be allocated to different nodes and then processed in parallel. Such divide-and-conquer approaches are particularly useful in settings where communication among different workers is costly or limited.
Moreover, privacy concerns, governance issues, data security or institutional constraints often make sharing data difficult, 
a situation which is common when processing medical data. We suggest a novel method for the distributed calculation of the marginal likelihood with large data sets under communication constraints.

Distributed Bayesian computation largely falls in the MapReduce or split-apply-combine frameworks for computation \citep{dean2008mapreduce, wickham2011split}. The distributed Bayesian computational workflow can be described in three key steps:
\begin{itemize}
    \item Split: Divide the data set into subsets and distribute across nodes (workers).
    \item Apply: Each worker independently computes a posterior distribution based on a subset.
    \item Combine: Aggregate the worker posterior distributions to form a consensus.
\end{itemize}
The combine step involves the synthesis of multiple probability distributions, and falls under the general umbrella of belief aggregation, a well studied topic in meta-analysis and probabilistic forecasting \citep{west_1984_aggregation, genest_1984_characterization, dietrich_2010_bayesian}. Our goal is to determine how marginal likelihood based Bayesian model choice can be operationalised within the aforementioned computational framework. We find that belief aggregation rules for model choice differ in some key aspects compared to fixed model inference and examine the theoretical and computational consequences.

In summary, our main contributions are the following. 
\begin{itemize}
    \setlength\itemsep{-0.2em}
    \item We derive a general decomposition of the marginal likelihood that enables efficient divide-and-conquer calculation 
    without accessing the data in one single place. 
    The combination of the results requires minimal communication and no exchange of data. 
    The computational complexity per worker is $\mathcal{O}(n/S)$ instead of $\mathcal{O}(n)$ on a single machine, 
    where $n$ is the number of observations and $S$ the number of workers. 
    \item We provide a theoretical analysis of two different algorithms for distributed calculation of the marginal likelihood. The first is a simulation consistent approach making use of data augmentation, and the second is an approximate approach relying on local normal approximations. Error bounds are developed for the approximate approach. 
    \item We illustrate the performance of our method on several challenging applications with millions of data points and show that the computation time is reduced by several orders of magnitude, 
    incurring only a negligible bias.
    \item We show how to apply our approach in a reversible jump setting where an MCMC sampler moves between different feature spaces. 
\end{itemize}

The rest of our work is structured as follows. We discuss related work in Section \ref{sec:relatedwork} and 
introduce relevant background material in Section \ref{sec:background}. Then we present our main contributions
on distributed Bayesian model choice in Section \ref{sec:distributedmodelchoice}. 
In Section \ref{sec:experiments} we demonstrate the applicability of our approach on several data sets and models 
before discussing possible extensions in Section \ref{sec:discussion}.

\section{Related Work} \label{sec:relatedwork} 
Previous work on computationally efficient methods for processing data sets with massive sample sizes can generally be divided into two broad categories. \citet{bardenet_markov_2017}, \citet{jahan2020survey} and \citet{zhu2017big} provide reviews focusing on Bayesian methods for big data. 

A first stream of work focuses on speeding up computation using mini batches (i.e., random subsets) of the entire data. 
Initially introduced via optimisation in the field of machine learning, the idea of approximating a posterior using 
mini batches has received substantial attention since the work of \citet{welling2011bayesian} on stochastic gradient Langevin sampling 
and the work of \citet{hoffman2013stochastic} on mini batch sampling for variational inference. 
Subsequent work on this family of algorithms has included theoretical analysis and practical extensions, see, for example, \citet{chen2014stochastic,alquier2016noisy,quiroz2019speeding,dang2019hamiltonian}.  
The calculation of normalising constants, however, has received less attention despite the work of \citet{lyne2015, gunawan2018subsampling}. 

The second line of work aims to make use of parallel processing for reducing computation time. The idea of using a divide-and-conquer approach has seen interest in both the statistics and the machine learning community, 
see, for example, \citet{deisenroth2015distributed} for an application in Gaussian processes and \citet{jordan2019communication} 
for a general approach under communication constraints. A variety of follow-up has been sparked by the work of \citet{scott2016bayes}. 
The consensus Monte Carlo (CMC) approach performs posterior sampling on data shards and combines the results on a 
single worker. This approach has been discussed and improved both from a theoretical and practical perspective 
\citep{wang2013parallelizing, neiswanger2013asymptotically, minsker14, scott2017comparing, srivastava2018scalable, zhang2018robust, szabo2019asymptotic}. 
The idea of distributed computation has since been picked up in different communities as, for example, in sequential Monte Carlo 
\citep{rendell2018global} and expectation propagation \citep{gelman2017expectation, barthelme2018divide}. 
Combining different models using different data sources has also received substantial attention, see for instance \citet{goudie2019joining, jacob2017better}. 
The concept of federated learning \citep{li2020federatedlearning}, where the estimation of a model is achieved in a highly distributed setting with repeated rounds of communication between a central node and workers, is another adjacent field to our work. Despite growing interest in Bayesian federated learning \citep{chen2020fedbe,yurochkin2019bayesian}, Bayesian model choice has seen little investigation. 

We will proceed under the assumption that there are barriers to implementing a mini-batch based algorithm. Storing the data in one location may be infeasible, or privacy restrictions and communication costs rule out the possibility for one central node to be connected to each site with a subset of the data. Situations where this may occur include the analysis of electronic health records and the analysis of genetic data from large cohorts. It is natural to consider the Laplace approximation to the marginal likelihood given the large sample size \citep{kass1995bayes}. However, the assumed barriers impede the ability to compute directly the Laplace approximation as determining the posterior mode and Hessian would require extensive communication between sites and the central node, see, for example, \citet{mcmahan2017communication, safaryan2021fednl}. Divide-and-conquer approaches will still be feasible in such a setting.
\section{Background}
\label{sec:background}

We first give an overview on Bayesian model choice using the marginal likelihood, and useful computational techniques for estimation of the marginal likelihood. We then discuss core principles for distributed Bayesian inference, and cover the use of normal approximations in the consensus Monte Carlo algorithm by \citet{scott2016bayes}.

\subsection{Bayesian model choice} \label{ssec:bmchoice} 
\paragraph{The normalising constant}
We define the posterior distribution of a parameter $\theta \in \Theta \subset \mathbb{R}^p $ given data $y$ as $ p(\theta| y) = p(y|\theta) p(\theta) / p(y)$. It depends on the unknown normalising constant $p(y)$. This normalising constant is the marginal likelihood of the data given the model, also called model evidence, and is calculated as 
$$p(y) = \int_\Theta p(y|\theta) p(\theta) d \theta. $$ 
In most settings this constant is not analytically tractable, as it involves the integration over a potentially high dimensional parameter space $\Theta$. 
Various sampling based methods are available for the approximate calculation of the evidence, for example, importance sampling (IS) \citep{geweke1989bayesian}, sequential Monte Carlo \citep{del2006sequential}, nested sampling \citep{skilling2006nested} or bridge sampling \citep{meng1996simulating, gelman1998simulating}. See \citet{Knuth2015} for a review of the different methods. 

\paragraph{The Bayes factor}
The posterior distribution over a set of competing models is defined using the marginal likelihood and a prior distribution on models. Models may be distinguished by different choices of link functions, hyper parameters or feature spaces. The posterior probability of model $m_i$ is given by
\begin{eqnarray} \label{eq:postprobmodel}
p(m_i | y) = \frac{p(y | m_i) p(m_i)}{ \sum_{k=1}^K p(y|m_k) p(m_k)}, 
\end{eqnarray}
where  $p(y | m_i) = \int_\Theta p(y|\theta, m_i) p(\theta|m_i) d \theta$, with both prior and likelihood now being dependent on the model $i$.  Models may also be compared using the Bayes factor (BF) \citep{kass1995bayes}.  The BF is calculated as 
\begin{eqnarray*} \label{eq:bayesfactor}
B_{m_1, m_2} =  \frac{p(y | m_1)}{p(y|m_2)}  = \frac{p(m_1 | y)}{p(m_2 | y)}  \times \frac{p(m_2)}{p(m_1)},
\end{eqnarray*}
where $p(m_i)$ denotes the prior probability of model $i$, $p(m_i|y)$ denotes the posterior probability of the model given the data and
$p(y|m_i)$ for $i=1,2$ denotes the probability of the data given the model. For two models the BF allows to select the model with highest posterior probability while adjusting for the prior odds. 
It is an alternative to standard statistical testing when it comes to model choice.

\paragraph{Data augmentation and conditionally conjugate models}  
Data augmentation has proven to be an important technique for Bayesian computation, facilitating both posterior sampling and marginal likelihood estimation \citep{meng_2001_data_augmentation, tanner_2010_data}. The rationale is to introduce a latent variable $z \in \mathcal{Z}$ such that the complete data model $p(y, z | \theta)$ is more mathematically and computationally tractable than the marginal model $p(y | \theta) = \int_{\mathcal{Z}} p(y, z, |\theta) \ dz$. Example applications include binary regression models, mixture models and factor analysers \citep{meng_2001_data_augmentation, holmes2006, tanner_2010_data}. A particularly useful strategy
is to construct a sampler on the extended space of $\theta, z \sim p(\theta, z | y)$ and to marginalise out $z$ once a cloud of samples has been generated. This can be an efficient strategy if the full conditionals $p(\theta | y, z)$ and $p(z | y, \theta)$ have a known distribution 
and allows the construction of a Gibbs sampler. Conditionally conjugate models are obtained when selecting the prior $p(\theta)$ to be conjugate to the complete data likelihood $p(y, z, | \theta)$ \citep{gelman_2013_bayesian}. Prior $p(\theta)$ and the conditional posterior $p(\theta| y, z)$ are part of the same distribution family and hence sampling from the conditional posterior becomes straightforward for commonly chosen priors. 
The posterior distribution can be represented as $p(\theta | {y}) = \int p({\theta}| z, y)p(z | y) \ dz$. For any ordinate $\theta$, the log marginal likelihood satisfies
\begin{align*}
    \log p(y) &= \log p(\theta) + \log p(y | \theta) - \log \int p({\theta}| z, y)p(z | y) \ dz
\end{align*}
Following Chib \citep{chib1995marginal}, a simulation consistent estimator of the model evidence is 
$$\log \widehat{p}(y) =  \log p(\theta) + \log p(y|\theta) - \log \sum_{i=1}^{N}p(\theta | z^{i}, y)/N, $$
where the $z^{i}$ are samples from the posterior distribution $p(z | y)$.

\subsection{Distributed Bayesian inference}
We assume that we observe data $y \in \mathcal{Y}$ consisting of $n$ data points that can be split into $S$ non overlapping data shards $y_s$, 
potentially containing several observations such that $y = \{ y_1, \cdots, y_S \}$. The likelihood is assumed to satisfy the independence condition $p(y |\theta)=\prod_{s=1}^{S}p(y_{s}|\theta)$. As outlined in the introduction, the general strategy for distributed Bayesian inference is to allocate a shard of data to each worker who will then compute a local posterior distribution on the basis of the subset $y_{s}$. The major challenge in this approach is the formation of a consensus probability distribution given the individual posterior distributions computed by each worker. 

An important tool for the synthesis of probability distributions is product-of-experts pooling \citep{genest_1984_characterization, hinton_2002_training, dietrich_2010_bayesian}. Taking $\theta$ to be the quantity of interest, the products-of-experts consensus distribution $q(\theta)$ is formed by taking the product of the worker distributions $\{q_{s}(\theta)\}_{s=1}^{S}$, so $q(\theta) \propto \prod_{s=1}^{S}q_{s}(\theta)$.  The significance of product-of-experts pooling is readily seen when the model is treated as being fixed and known. The full posterior can be represented as
\begin{eqnarray} \label{eq:posteriordecomp1}
p(\theta | y) \propto \prod_{s=1}^S p(y_s|\theta) p(\theta)^{1/S},
\end{eqnarray}
where $p(y_s|\theta)$ is a likelihood factor over the shard $y_s$ and $p(\theta)^{1/S}$ is the unnormalised subprior, i.e. a fraction of the initial prior $p(\theta)$. Each worker may compute a so-called subposterior $\widetilde{p}(\theta | y_{s})$ on a shard of data using the fractionated prior $\tilde{p}(\theta | y_s) \propto  p(y_s|\theta) p(\theta)^{1/S}$. The full posterior distribution can  then be obtained by product-of-experts pooling of the subposteriors \citep{huang_2005_sampling, scott2016bayes}
\begin{eqnarray} \label{eq:posteriordecomp2}
p(\theta | y) \propto \prod_{s=1}^S \tilde{p}(\theta|y_s).
\end{eqnarray} 
This important property does not transfer to model selection. There is no analogous product-of-experts decomposition of the full posterior ditribution over models
\begin{align} \label{eq:poe_models}
    p(m_{i} | y) \not\propto \prod_{s=1}^{S}\tilde{p}(m_{i} | y_{s}),
\end{align}
where the subposterior model probabilities are defined by $\widetilde{p}(m_{i} | y_{s}) \propto \widetilde{p}(y | m_{i})p(m_{i})^{1/S}$ and the subposterior evidence $\widetilde{p}(y_{s} | m_{i}) = \int_\Theta p(y_{s}|\theta, m_{i})\tilde{p}(\theta | m_{i}) \ d\theta$ is determined using the normalised subprior $\widetilde{p}(\theta|m_{i})=p(\theta|m_{i})^{1/S}/\int_\Theta p(\theta |m_{i})^{1/S} d\theta$. This is due to the fact that in general,
$$
p(y|m_{i}) = \int_\Theta \prod_{s=1}^S p(y_s|\theta, m_{i}) p(\theta |m_{i})^{1/S} d \theta \neq \prod_{s=1}^S \int_\Theta p(y_s|\theta, m_{i}) p(\theta | m_{i})^{1/S} d \theta.
$$
The lack of a product-of-experts decomposition for the full posterior on models \eqref{eq:poe_models} suggests that the protocol for distributed Bayesian inference in the fixed model setting may not be effective for model choice. If the goal is to recover the full posterior model probabilities $p(m_{i} | y)$, it is no longer sufficient to compute subposterior distributions over models in the apply step and then form a consensus distribution using product-of-experts pooling in the combine step. The correct belief aggregation procedure for model selection is developed in Section 4. 

In practice, each worker will typically return a Monte Carlo or analytic approximation of the subposterior distribution. An important consideration with distributed Bayesian inference is how the mean squared error of the consensus scales with the number of subsets $S$ \citep{bardenet_markov_2017, scott2016bayes, neiswanger2013asymptotically}. There is almost always a trade-off between the error and the computational benefits afforded by distributing the workload across more nodes $S$. This dynamic will be a key feature in our theoretical analysis and computational experiments. 
\subsection{Consensus Monte Carlo}
The Consensus Monte Carlo algorithm \citep{scott2016bayes} is based on the product-of-experts posterior decomposition \eqref{eq:posteriordecomp2}. The key to the approach is that the product of normal subposteriors is proportional to another normal distribution:
$$ \prod_{s=1}^S \mathcal{N}(\theta | \mu_s, \Sigma_s) \propto \mathcal{N}(\theta | \mu, \Sigma), $$ 
where the overall variance and mean are obtained as $\Sigma^{-1} = \sum_{s=1}^S \Sigma_s^{-1} $ and 
$\mu = \Sigma \sum_{s=1}^S \Sigma_s^{-1} \mu_s$. 

The consensus Monte Carlo algorithm samples $N$ points $\theta_s^1, \cdots, \theta_s^N$ from each of the individual subposteriors $ \tilde{p}(\theta|y_s)$. 
This sampling can be achieved, for example, using standard MCMC algorithms like random walk Metropolis-Hasting \citep{hastings1970monte, 10.1093/biomet/asz066} or Hamiltonian Monte Carlo \citep{neal2011mcmc}. 
The samples are then recombined using a normal approximation to the subposteriors. 

The normal approximation $\mathcal{N}(\mu_s, \Sigma_s) \approx \tilde{p}(\theta|y_s)$ is based on the estimated mean $\mu_s$ and variance $\Sigma_s$ from the subposterior samples, 
using the Laplace-Metropolis approximation \citep{lewis1997estimating}. 
This approximation is asymptotically justified through the Bernstein-von-Mises (BvM) theorem (see also \citet{Ghosh:1608771}). %
The recombination of the samples from the local Markov chains is based on a weighting according to the inverse covariances of the subposteriors \citep{scott2016bayes}. Combining the sampling based approach with the normal approximation has the advantage that more features of the posterior distribution are captured compared to the use of a plain normal approximation where the sampling step would be omitted. See \citet{scott2016bayes, scott2017comparing} for more details. %

\section{Distributed Model Choice} \label{sec:distributedmodelchoice}
We will now introduce and discuss our contributions making use of the previously introduced background material. 
\paragraph{Decomposing the model evidence}
Using Bayes' Theorem and some elementary algebra, it is possible to obtain an identity for the marginal likelihood that lends itself to distributed computation. 
\begin{proposition} \label{thm:prop1}
The model evidence for the full data can be decomposed as 
\begin{eqnarray} \label{eq:evidencedecomp}
p(y) = \alpha^S \prod_{s=1}^S \tilde{p}(y_s) \int_\Theta \prod_{s=1}^S \tilde{p}(\theta|y_s) \ d \theta , 
\end{eqnarray}
where $\tilde{p}(\theta) = \frac{p(\theta)^{1/S}}{\alpha}$ is the normalised subprior, $\alpha_\Theta = \int p(\theta)^{1/S} d \theta$ is the normalising constant of the subprior, $\tilde{p}(y_s) = \int_\Theta p(y_s|\theta) \tilde{p}(\theta)d \theta$ is the normalising constant of the subposterior and $\tilde{p}(\theta|y_s) = p(y_s|\theta) \tilde{p}(\theta)/\tilde{p}(y_s)$ denotes the normalised subposterior. 
\end{proposition}
\begin{proof}
See Appendix.
\end{proof}
Each of the three components in the decomposition \eqref{eq:evidencedecomp} has an interpretation in terms of the generic split-apply-combine framework mentioned in the introduction. The term $\alpha^{S}$ reflects that each worker is allocated a fraction of the prior information in the split step. The subposterior evidence for each shard $\tilde{p}({y}_{s})$ is computed by the workers in the apply stage. The evaluation of the subposterior integral $\int_{\Theta} \prod_{s=1}^{S}\tilde{p}({\theta} | y_{s}) d \theta $ is a necessary step for appropriate evidence synthesis in the combine stage.

Using Proposition \ref{thm:prop1} it is straightforward to show that the full posterior distribution on models has the modified product-of-experts representation
\begin{align} \label{eq:poe_models_adjusted}
    p(m_{i} | y) \propto \left\lbrace \prod_{s=1}^{S}\tilde{p}(m_{i} | y_{s}) \right\rbrace \alpha_{i}^{S}\left( \int_\Theta \prod_{s=1}^S \tilde{p}(\theta|y_s, m_{i})\  d \theta\right),
\end{align}
recalling that the subposterior probabilities are given by $\widetilde{p}(m_{i} | y_{s}) \propto \tilde{p}(y_{s} |m_{i})p(m_{i})^{1/S}$. Comparing \eqref{eq:poe_models_adjusted} to  \eqref{eq:poe_models}, we see that a modified product-of-experts rule is necessary for distributed Bayesian model selection, as the subposterior model probabilities $\tilde{p}(m_{i} | y_{s})$ do not contain enough information to recover the full posterior model probabilities $p(m_{i} | y)$. Subset analyses will be under-powered relatively compared to the full data set analysis if the shard size is small compared to the total sample size, and this may manifest in  a bias towards smaller models in the model subposteriors $\tilde{p}(m_{i} | y_{s})$. A secondary issue is that a model may appear to fit well in each subset but be of poor fit overall. The inclusion of the integral term over the subposteriors $ \int_\Theta \prod_{s=1}^S \tilde{p}(\theta|y_s, m_{i})\  d \theta$ allows the global goodness of fit to be reconstructed by considering the overlap in the subposterior distributions. 

We will now focus on how to use the presented decomposition in an algorithm to approximate the marginal likelihood of the full data set. In a variety of settings $\alpha$ can be computed analytically. For example, if $p(\theta)$ is a normal distribution, 
fragmenting the prior amounts to an inflation of the prior variance:
$\mathcal{N}(0, \Sigma)^{1/S} \propto \mathcal{N}(0, S \Sigma)$ and $\alpha$ is thus obtained easily. 
The same holds true for a Laplace prior, where $\mathcal{L}(0, \sigma)^{1/S} \propto \mathcal{L}(0, S \sigma)$. 
However, care must be taken for some distributions of the exponential family. When $S$ is too large, 
the integral $\int_\Theta p(\theta)^{1/S} d \theta$ becomes infinite as too much mass is pushed into the tails. 
We recommend checking this on a case by case basis. 
The local evidence $\tilde{p}(y_s)$ can be calculated by every individual worker using one of the various techniques described 
at the start of Section \ref{ssec:bmchoice}. 
The most difficult part of \eqref{eq:evidencedecomp} is to evaluate the last integral 
\begin{eqnarray} \label{eq:isubbase}
    I_{\text{sub}} := \int_\Theta \prod_{s=1}^S \tilde{p}(\theta|y_s) \  d \theta.
\end{eqnarray}
We will now discuss two different strategies to estimate the integral $I_{\text{sub}}$ \eqref{eq:isubbase}, an exact approach based on data augmentation and Gibbs sampling, and an approximate approach using normal approximations. 

\subsection{Data augmentation to compute $I_{\text{sub}}$} 
A data augmentation strategy can be used to construct a simulation consistent estimator of the components in \eqref{eq:evidencedecomp}. See \citet[Chap 2]{ahfock2019thesis} for a detailed discussion. In particular, the latent variables are helpful to calculate an estimator of \eqref{eq:isubbase} as we show now. %
We assume that the latent variables $z_s$ are independent given $\theta$. Then, the augmented full data set posterior is again proportional to the subposteriors, i.e., 
\begin{eqnarray*}
     p(\theta | y_{1:S}, z_{1:S}) &\propto& p(\theta) p(y_{1:S}, z_{1:S} | \theta) \propto \prod_{s=1}^S p(y_{s}, z_{s} | \theta) p(\theta)^{1/S}, \\
     &\propto& \prod_{s=1}^S p(y_{s}, z_{s} | \theta) \tilde{p}(\theta) \propto \prod_{s=1}^S  \tilde{p}(\theta| y_{s}, z_{s}). 
\end{eqnarray*}
The conditional latent variable subposteriors on the individual workers (accessing only $y_s$) are given as as 
\begin{eqnarray*}
    \tilde{p}(z_{s}| y_{s}, \theta) \propto p(y_{s}, z_{s}| \theta) \tilde{p}(\theta) \propto  p(z_{s} | y_{s}, \theta) p(y_{s}| \theta) \propto p(z_{s}| y_{s}, \theta),
\end{eqnarray*}
and hence the subprior $\tilde{p}(\theta)$ is not material when conditioning on $\theta$.
This establishes a Gibbs sampling strategy for sampling from $\tilde{p}(\theta, z_{s}| y_{s})$ using the conditionals defined above. 
In conditionally conjugate models, where the subpriors are chosen such that $\int_\Theta \prod_{s=1}^S \tilde{p}(\theta|y_s, z_s) d \theta$ has a closed form solution, we 
can exploit the samples generated from the latent variable subposteriors to approximate $I_\text{sub}$. 
\begin{proposition} \label{thm:propisub}
Using a data augmentation scheme for the augmented data likelihood $p(y_{1:S}, z_{1:S} | \theta)$ we have
\begin{eqnarray} \label{eq:isub}
    I_{\rm{sub}} = \mathbf{E}_{\tilde{p}(z_{1:S}|y_{1:S})} \left[\int_\Theta \prod_{s=1}^S \tilde{p}(\theta|y_s, z_s) d \theta \right],
\end{eqnarray}
where $\tilde{p}(z_{1:S}|y_{1:S}) = \prod_{s=1}^S \tilde{p}(z_s|y_s)$. 
\end{proposition}
\begin{proof}
See Appendix.
\end{proof}
As a a consequence, Proposition \ref{thm:propisub} suggests the following Monte Carlo estimator
\begin{eqnarray} \label{eq:isubestimator}
\hat{I}_{\text{sub}} = \frac{1}{N} \sum_{i=1}^N \int_\Theta \prod_{s=1}^S \tilde{p}(\theta|y_s, z_s^i) d \theta,
\end{eqnarray}
where $\int_\Theta \prod_{s=1}^S \tilde{p}(\theta|y_s, z_s^i) d \theta$ is calculated analytically and $z_s^i$ are 
samples obtained from, e.g., a Gibbs sampler. See the algorithmic version of the sampler in Algorithm \ref{alg:exactmethod}. 
We provide further details on data augmentation for logistic regressions using P\'{o}lya-Gamma data augmentation in the appendix. 
The variance of the proposed estimator can be understood through a connection to importance sampling. 
\begin{proposition} \label{thm:propisubvar}
Assume i.i.d. sampling from the latent variable posterior. The variance of the estimator suggested in \eqref{eq:isubestimator} is then
\begin{eqnarray*} 
\Var \hat{I}_{\rm{sub}} = \frac{I_{\rm{sub}}^2}{N} 
\Var_{\tilde{p}(z_{1:S}|y_{1:S})} \left[ \frac{p(z_{1:S}|y_{1:S})}{\prod_{s=1}^S \tilde{p}(z_s|y_s)} \right].
\end{eqnarray*}
\end{proposition}
\begin{proof}
See Appendix.
\end{proof}
    
This links the variance of the estimator $\hat{I}_{\text{sub}}$ to the relative coverage of the joint conditional posterior $p(z_{1:S}|y_{1:S})$ 
and the product of the conditional subposteriors $\tilde{p}(z_{1:S}|y_{1:S})$. The variance of the ratio
measures the quality of the product of the conditional subposterior as an importance sampling proposal distribution. If the tails of the proposal are thinner than the target, the variance of the estimator potentially becomes infinite. As the number of splits increases, we expect the product of the conditional subposteriors to become a worse approximation of the true subposterior, and the variance of the estimator to increase. %
Note that in practice the variance of $\hat{I}_{\text{sub}}$ will be larger than the suggested quantity due to the autocorrelation introduced by the Gibbs sampler. 

\begin{algorithm}[ht]
    \caption{Distributed model evidence computation using data augmentation}
    \label{alg:exactmethod}
 \begin{algorithmic}
    \STATE {\bfseries Input:} data $y$, number of chunks $S$, likelihood $p(\cdot|\theta)$, prior $p(\theta)$.
    \STATE {\bfseries Split}: Divide data in $S$ chunks $y_1, \cdots, y_S$.
    \STATE {\bfseries Apply} in parallel:
    \FOR{$s=1$ {\bfseries to} $S$}
    \STATE Sample $(\theta_s^1, z_s^1), \cdots, (\theta_s^N, z_s^N) \sim \tilde{p}(\theta, z | y_s)$ using Gibbs sampling. 
    \STATE Calculate and store $\tilde{p}(y_s) = \int p(y_s|\theta) \tilde{p}(\theta) d \theta$ using Chib's method.
    \ENDFOR
    \STATE {\bfseries Combine:}
    \STATE Calculate $ p(y) = \alpha^S \prod_{s=1}^S \tilde{p}(y_s) \int_\Theta \prod_{s=1}^S \tilde{p}(\theta| y_s) d \theta  $, 
    \STATE using 
    $\int_\Theta \prod_{s=1}^S \tilde{p}(\theta| y_s) d \theta = \mathbf{E}_{\tilde{p}(z_{1:S} | y_{1:S})} \left(  \int_{\Theta} \prod_{s=1}^S \tilde{p}(\theta|y_s, z_s) d \theta \right) \approx \eqref{eq:isubestimator}$.
 \end{algorithmic}
 \end{algorithm}

\subsection{A normal approximation to $I_{\text{sub}}$} \label{ssec:approximateapproach} 
Models relying on data augmentation are a rather restricted class of models that require tailor-made samplers. A more widely applicable approach to estimate $I_{\text{sub}}$ \eqref{eq:isubbase} is a normal approximation to each subposterior $\widetilde{p}(\theta | y_{s})\approx \mathcal{N}(\theta | \mu_{s}, \Sigma_{s})$, yielding the following approximation of the integral:
\begin{align*}
   I_{\text{sub}} &= \int_{\Theta} \prod_{s=1}^{S} \tilde{p}({\theta} | y_{s}) \ d\theta \approx \int_{\Theta} \prod_{s=1}^{S}\mathcal{N}({\theta}\mid {{\mu}}_{s}, {{\Sigma}}_{s}) \ d{\theta}.
\end{align*}
The mean and covariance parameters for the subposterior normal approximations, $\mu_{s}, \Sigma_{s}$ can be estimated given subposterior samples generated by each worker. Using the fact that the product of normal density functions is proportional to another normal density we can provide the closed form expression
\begin{eqnarray} \label{eq:gaussianisub}
\int_{\Theta} \prod_{s=1}^S \mathcal{N}(\theta| \mu_s, \Sigma_s) \ d \theta
= \exp \left(\sum_{s=1}^S \xi_s - \xi \right), 
\end{eqnarray}
where 
\begin{eqnarray*}
\eta_s = \Sigma_s^{-1} \mu_s, \Lambda_s = \Sigma_s^{-1}, \\ \xi_s = - \frac{1}{2}(p \log 2 \pi - \log |\Lambda_s| + \eta_s^t\Lambda_s \eta_s), \\
\eta = \sum_{s=1}^S \Sigma_s^{-1} \mu_s, ~ \Lambda = \sum_{s=1}^S \Sigma_s^{-1}, \\ \xi = - \frac{1}{2}(p \log 2 \pi - \log |\Lambda| + \eta^t\Lambda \eta).
\end{eqnarray*}
The normal approximations to the subposterior $\widetilde{p}(\theta | y_{s})\approx \mathcal{N}(\theta | \mu_{s}, \Sigma_{s})$ can be justified under the usual conditions for a Bernstein-von-Mises theorem to hold, see e.g. \citet{Ghosh:1608771}. Due to the involved matrix inversions and calculation of determinants the complexity of calculating \eqref{eq:gaussianisub} is $O(S p^3)$. 
We approximate the normalising constants $ \tilde{p}(y_s)$ using any of the techniques discussed at the beginning of Section \ref{ssec:bmchoice}. Note that the normal approximation is only required for estimating $I_{\text{sub}}$, and does not influence the estimation of $\widetilde{p}(y_{s})$. We hence suggest Algorithm \ref{alg:approxmethod} for the estimation of the marginal likelihood. 

\paragraph{Using mixtures of normals}
If the posterior distribution is multimodal, a simple normal approximation will result in a poor approximation of the integral of the product of the subposteriors. 
As noted by \citet{neiswanger2013asymptotically, scott2017comparing}, a mixture of normals or kernel density estimators can be used to approximate the posterior. 
These approaches have been shown to work well in practice. However, if the number of mixture components or the 
number of splits $S$ gets large this approach may become prohibitive. As we have to calculate the product of $S$ mixtures, 
the resulting calculation of the product of distributions has a complexity of $\OO(S^k )$, where $k$ denotes the number of mixture components. %
We therefore refrain from this approach. 

\begin{algorithm}[t!]
   \caption{Distributed model evidence computation using normal approximations}
   \label{alg:approxmethod}
\begin{algorithmic}
   \STATE {\bfseries Input:} data $y$, number of chunks $S$, likelihood $p(\cdot|\theta)$, prior $p(\theta)$.
   \STATE {\bfseries Split}: Divide data in $S$ chunks $y_1, \cdots, y_S$.
   \STATE {\bfseries Apply} in parallel:
   \FOR{$s=1$ {\bfseries to} $S$}
   \STATE Sample $\theta_s^1, \cdots, \theta_s^N \sim \tilde{p}(\theta|y_s)$.
   \STATE Calculate and store $\tilde{p}(y_s) = \int p(y_s|\theta) \tilde{p}(\theta) d \theta$ and $\mu_s, \Sigma_s$.
   \ENDFOR
   \STATE {\bfseries Combine:}
   \STATE Calculate $\alpha^S \prod_{s=1}^S \tilde{p}(y_s) \int  \prod_{s=1}^S \mathcal{N}(\theta| \mu_s, \Sigma_s) d \theta  \approx p(y)$.
\end{algorithmic}
\end{algorithm}

\paragraph{Error of the approximation}
Proposition \ref{prop:isub_normal} gives an exact representation of the relative error when using the normal approximation to the subposterior integral $I_{\text{sub}}$ \eqref{eq:gaussianisub}.
\begin{proposition}\label{prop:isub_normal}
Suppose $\{ \mu_{s}, \Sigma_{s} \}_{s=1}^{S}$ are the selected parameters for the subposterior normal approximations $\tilde{p}(\theta | y_{s})\approx \mathcal{N}(\theta | \mu_{s}, \Sigma_{s})$, and that $\mathcal{N}(\theta | \mu, \Sigma) \propto \prod_{s=1}^{S} \mathcal{N}(\theta | \mu_{s}, \Sigma_{s})$ is the resulting normal approximation to the full posterior $p(\theta | y)$. Then the proposed normal approximation to $I_{\rm{sub}}$ satisfies:
\begin{align*}
      \int_{\Theta} \prod_{s=1}^{S}\widetilde{p}({\theta}\mid {y}_{s}) \ d\theta &= \left\lbrace  {{\int_{\Theta}  \prod_{s=1}^{S}\mathcal{N}({\theta}\mid {{\mu}}_{s}, {{\Sigma}}_{s}) \ d{\theta}}} \right\rbrace {{\mathbb{E}_{\mathcal{N}({\theta}; {{\mu}}, {{\Sigma}})}\left[ \prod_{s=1}^{S}\dfrac{\widetilde{p}({\theta} \mid {y}_{s})}{\mathcal{N}({\theta}| {{\mu}}_{s}, {{\Sigma}}_{s})} \right]}}.
    \end{align*}
\end{proposition}
\begin{proof}
See Appendix.
\end{proof}
 Proposition \ref{prop:log_error} provides an error bound for the proposed approximation to the marginal likelihood under assumptions on the quality of the subposterior normal approximations. 
\begin{proposition}\label{prop:log_error}
Suppose $\{ \mu_{s}, \Sigma_{s} \}_{s=1}^{S}$ are the selected parameters for the subposterior normal approximations $\tilde{p}(\theta | y_{s})\approx \mathcal{N}(\theta | \mu_{s}, \Sigma_{s})$, and that $\mathcal{N}(\theta | \mu, \Sigma) \propto \prod_{s=1}^{S} \mathcal{N}(\theta | \mu_{s}, \Sigma_{s})$ is the resulting normal approximation to the full posterior $p(\theta | y)$. The exact marginal likelihood $p(y)$ and the proposed approximation to the marginal likelihood $\widehat{p}(y)$ are respectively,
   \begin{align*}
          {p}({y}) &= {{\alpha^{S}}}{{\left(\prod_{s=1}^{S}\widetilde{p}({y}_{s})\right)}}{{\int_{\Theta} \prod_{s=1}^{S}{\tilde{p}({\theta}\mid y_{s})} \ d{\theta}}},  \\
       \widehat{p}({y}) &= {{\alpha^{S}}}{{\left(\prod_{s=1}^{S}\widetilde{p}({y}_{s})\right)}}{{\int_{\Theta} \prod_{s=1}^{S}{\mathcal{N}({\theta}\mid {\mu}_{s}, {\Sigma}_{s})} \ d{\theta}}}. 
    \end{align*}
    Suppose that the subposterior normal approximations satisfy the density ratio bounds
        \begin{align*}
            \max_{s=1, \ldots, S} \ \sup_{{\theta} \in {\Theta}} \dfrac{\widetilde{p}({\theta} \mid {y}_{s})}{\mathcal{N}({\theta}; {\mu}_{s}, {\Sigma}_{s})}  \le A, \qquad      \max_{s=1, \ldots, S} \ \sup_{{\theta} \in {\Theta}} \dfrac{\mathcal{N}({\theta}; {\mu}_{s}, {\Sigma}_{s})}{\widetilde{p}({\theta} \mid {y}_{s})} \le B.
        \end{align*}
Then the relative error of the proposed approximation to the marginal likelihood satisfies
\begin{align*}
            -S \log B + \log \mathbb{E}_{\mathcal{N}({\theta}; {\mu}, {\Sigma})}[\mathbbm{1}({\theta} \in \Theta)] \le \log \dfrac{p({y})}{\widehat{p}({y})} \le S \log A.
\end{align*}
\end{proposition}
\begin{proof}
See Appendix.
\end{proof}
If the parameter space is unconstrained, so $\Theta = \mathbb{R}^{\text{dim}(\theta)}$, then $ \mathbb{E}_{\mathcal{N}({\theta}; {\mu}, {\Sigma})}[\mathbbm{1}({\theta} \in \Theta)]=1$ and the bound can be simplified to   $-S \log B  \le \log {p({y})}/{\widehat{p}({y})} \le S \log A$. The supremum of the density ratio is an interesting divergence measure for two probability distributions, with important connections to the total variation distance \citep{dumbgen_2021_bounding}. Under the conditions of the Bernstein-von-Mises theorem, the quality of each subposterior normal approximation is expected to improve as the shard size increases. With a fixed number of shards $S$,  $A$ and $B$ will tend to $1$ as $n$ increases if $ \{\mu_{s}, \Sigma_{s}\}_{s=1}^{S}$ are taken to be the true subposterior means and variances. Corollary \ref{cor:growth} provides an error bound under an assumption about the average accuracy of the subposterior normal approximations over the subsets $S$, rather than the worst case accuracy over the subsets $S$, as in Proposition \ref{prop:log_error}. 
\begin{corollary}\label{cor:growth}
Fix the shard size $n_{s}$, and let the total sample size be given by $n=Sn_{s}$. Let $p(y)$ represent the true marginal likelihood and $\widehat{p}(y)$ denote the proposed approximation as in Proposition \ref{prop:log_error}.  Assume that 
\begin{align*}
    \sup_{{\theta} \in {\Theta}}  \left\lvert \dfrac{1}{S}\sum_{s=1}^{S}\log \dfrac{\widetilde{p}(\theta | y_{s})}{\mathcal{N}(\theta | \mu_{s}, \Sigma_{s})} \right\rvert &= \OO_{p}(1)
\end{align*}
as the number of shards $S$ increases, where $\{y_{s}, \mu_{s},\Sigma_{s}\}_{s=1}^{S}$ are treated as random variables. Then the relative error satisfies
\begin{align*}
    \log \dfrac{p({y})}{\widehat{p}({y})} &= \OO_{p}(S).
\end{align*}
\end{corollary}
\begin{proof}
See Appendix.
\end{proof}

Proposition \ref{prop:log_error} and Corollary \ref{cor:growth} suggest a trade-off between the computational gains from the proposed distributed strategy via the number of splits $S$ and the resulting approximation error.

In practice we transform \eqref{eq:evidencedecomp} to the log domain and get
$$
\log \widehat{p}(y) = S \log \alpha + \sum_{s=1}^S \log \tilde{p}(y_s) + \log \int \prod_{s=1}^S \mathcal{N}(\theta \mid \mu_{s}, \Sigma_{s}) d \theta.
$$
The corresponding estimator is a sum of $S+2$ terms so the variance of this quantity grows linearly in $S$. 
Moreover, when estimating $\log \tilde{p}(y_s)$ by $N$ samples from a Markov chain the MSE is typically of order $\mathcal{O}(1/N)$. 
Thus, more simulations can reduce this error. 
In summary, our results suggest that if the error that comes from MCMC sampling is relatively small and that the shard sizes are large enough so that the quality of the subposterior normal approximation is reasonable, our suggested approach will result in good approximations of the full data set marginal likelihood.

\subsection{Model choice using reversible jump} 
Calculation of the marginal likelihood for every model in the candidate set is only feasible when there are a small number of models under consideration. In many practical settings, the aim is often to choose among a large number of models, that potentially have different support. An important example is variable selection, where a particular model consists of a specific combination of selected variables. 

The reversible jump approach \citep{green1995reversible} allows the construction of a Markov Chain that jointly traverses models and the associated parameter spaces. The posterior probability of a model given the data $p(m_{i}|y)$, see \eqref{eq:postprobmodel}, is obtained by the relative time the sampler 
spends exploring that model. 

\paragraph{Distributed RJMCMC}
The modified product-of-experts belief aggregation rule for models \eqref{eq:poe_models_adjusted}
\begin{align*} 
    p(m_{i} | y) \propto \left\lbrace \prod_{s=1}^{S}\tilde{p}(m_{i} | y_{s}) \right\rbrace \alpha_{i}^{S}\left( \int_\Theta \prod_{s=1}^S \tilde{p}(\theta|y_s, m_{i})\  d \theta\right)
\end{align*}
illuminates how one may conduct distributed Bayesian model selection using reversible jump methodology. During the apply stage, it is sufficient for workers to compute subposterior model probabilities $\widetilde{p}(m_{i} | y_{s}) \propto \tilde{p}(y_{s} | m_{i})p(m_{i})^{1/S}$, and model subposteriors $\tilde{p}(\theta | y_{s}, m_{i}) \propto p(y_{s} | \theta, m_{i})p(\theta |m_{i})^{1/S}$. The subposterior integral $ \int_\Theta \prod_{s=1}^S \tilde{p}(\theta|y_s, m_{i})\  d \theta$ can be evaluated in the combine stage, and the term $\alpha_{i}^{S}$ can be  determined in the split stage. The belief aggregation rule \eqref{eq:poe_models_adjusted} can then be used to reconstruct the full posterior model probabilities.   Once again, a normal approximation to each subposterior $\tilde{p}(\theta|y_s, m_i) \approx \mathcal{N}(\theta | \mu_{s}^{i}, \Sigma_{s}^{i})$  can be used to approximate the integral $$ \int_\Theta \prod_{s=1}^S \tilde{p}(\theta|y_s, m_{i})\  d \theta \approx \int_\Theta \prod_{s=1}^S  \mathcal{N}(\theta | \mu_{s}^{i}, \Sigma_{s}^{i}) d \theta.$$ 

Reversible jump is particularly well suited to to the decomposition \eqref{eq:poe_models_adjusted} as the sampler simultaneously explores models and the associated parameter spaces. As such, estimates of $\tilde{p}(m_{i}|y_{s}), \mu_{s}^{i}$ and $\Sigma_{s}^{i}$ will be readily available from the RJMCMC output of each worker. Consequently, the splitting approach can effectively be combined with a reversible jump algorithm as we suggest in Algorithm \ref{alg:rjmcmcmethod}. See also the 
appendix for more details. 

\begin{algorithm}[ht]
    \caption{RJMCMC distributed model choice}
    \label{alg:rjmcmcmethod}
 \begin{algorithmic}
    \STATE {\bfseries Input:} data $y$, number of chunks $S$, set of models $m_k$ for $k=1,\cdots,K$, likelihoods $p(\cdot|\theta, m_k)$, priors $p(\theta| m_k), p(m_k)$.
    \STATE {\bfseries Split}: Divide data in $S$ chunks $y_1, \cdots, y_S$.
    \STATE {\bfseries Apply} in parallel:
    \FOR{$s=1$ {\bfseries to} $S$}
    \STATE Sample $\theta_s^1, \cdots, \theta_s^N|m_k \sim \tilde{p}(\theta| y_s, m_k)$ using RJMCMC over all $k$. 
    \STATE Calculate and store $\tilde{p}(m_k|y_s)$, $\mu_{s}^{k},\Sigma_{s}^{k}$ for models $k=1, \ldots, K$.
    \ENDFOR
    \STATE {\bfseries Combine:}
    \STATE Calculate $\widehat{p}(m_{k} | y) \propto \left\lbrace \prod_{s=1}^{S}\tilde{p}(m_{k} | y_{s}) \right\rbrace \alpha_{k}^{S}\left( \int_\Theta \prod_{s=1}^S \mathcal{N}(\theta| \mu_{s}^{k}, \Sigma_{s}^{k})\  d \theta\right)$. \\
    Bayes factors can be estimated as $\widehat{p}(m_{k}|y)p(m_{k'})/\{\widehat{p}(m_{k'}|y)p(m_{k})\}$ for any two models $(k,k')$.
 \end{algorithmic}
 \end{algorithm}

There are a number of known issues with RJMCMC that may limit its effectiveness within a divide-and-conquer approach. Reversible jump samplers are known to be hard to tune and often slow to converge. The construction of suitable proposal distributions can also be very challenging when models are not nested. Moreover, the models of interest have to have been visited a sufficient number of times to get reliable estimates. All these combined make the use of RJMCMC burdensome. With regards to distributed computation, an issue is that if the space of potential models is large, not all models of interest might have been explored on every data shard.

\section{Experiments} \label{sec:experiments} 
In our experiment section we investigate the following questions. 
(a) Are naive voting strategies, upsampled likelihoods and CMC based importance sampling an alternative to our approach?
(b) How does the approach based on data augmentation and the approach based on normal approximations perform relatively? 
(c) What is the magnitude of the error introduced from approximating $I_{\text{sub}}$ compared to the overall scale of the log marginal likelihood?
(d) How does the error of $I_{\text{sub}}$ behave with respect to the number of splits?
(e) What are practical gains from using the distributed approach on very large data sets?
(f) How reliable is the distributed approach for RJMCMC?
We answer these questions by assessing our approach with a Bayesian logistic regression, a normal linear regression and a linear regression with Laplace priors. 

Our overall experimental set up is the following: 
prior variances on the model parameters are set to $1$ and their means to $0$. We use a randomised splitting procedure. This means for $S$ splits we will have roughly $n/S$ samples per split, 
where we use a uniform sampling scheme without replacement, if not otherwise stated. We run every sampler $20$ times where at every 
iteration the splitting and the Monte Carlo sampler are initialised with a different seed. Thus, the observed variation in the outcome is a 
combination of the variation through splitting (i.e. different partitions) and Monte Carlo sampling. 
The number of generated MCMC samples per chain is $10,000$ where the first $2,000$ samples are discarded as burn-in. 
We assess the error of the estimation of the log marginal likelihood by comparing it to 
the result of the estimation on the whole data set, if this is computationally feasible in a reasonable amount of time. 

We use the RMSE (root mean squared error), defined as $$\sqrt{\text{MSE}} = \sqrt{ \mathbf{E}\| \log p(y) - \log \hat{p}(y)\|^2},$$ where $\hat{p}(y)$
is approximated using our suggested decomposition. The relative RMSE is defined as 
$\%\sqrt{\text{MSE}} = \frac{ \sqrt{ \mathbf{E}\| \log p(y) - \log \hat{p}(y)\|^2}}{\log p(y)} \times 100$. 
The use of the relative RMSE is justified through the scale of the actual quantity we are trying to estimate. 

\paragraph{Practical considerations}
We implement our algorithm using two generic \textit{R} packages, namely \textit{rstan} \citep{rstan2017} and 
the \textit{bridgesampling} package \citep{gronau2017bridgesampling}. 
\textit{rstan} allows convenient sampling from the posterior distribution of a model using HMC. 
The package can handle a variety of different models and the user has to provide only a simple code that describes the model. 
Tuning of hyper parameters and convergence checking is handled automatically. The \textit{bridgesampling} package can use \textit{rstan} 
models to calculate an approximation of the model evidence. 
For the reversible jump illustration we use the \textit{R2BGLiMS} package\footnote{https://github.com/pjnewcombe/R2BGLiMS}. 
This package performs model choice using reversible jump MCMC on logistic, normal and Weibull regression models. 
We explicitly run the experiment on single core architectures to illustrate the advantages of distributing the computation over a large number of small workers. 
Code for reproducing the results is available through the first author's github repository\footnote{https://github.com/alexanderbuchholz/distbayesianmc}.

\begin{figure}
    \begin{center}
    \begin{subfigure}[b]{0.49\textwidth}
\centerline{\includegraphics[width=0.8\columnwidth]{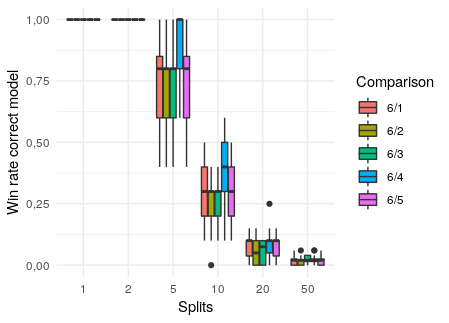}}
\caption{\label{fig:naive_local_bf} Majority based voting using local marginal likelihoods}
\end{subfigure}
\begin{subfigure}[b]{0.49\textwidth}
    \centerline{\includegraphics[width=0.8\columnwidth]{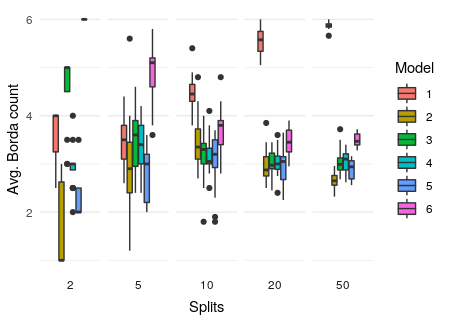}}
    \caption{\label{fig:borda_count} Borda counts using local marginal likelihoods}
\end{subfigure}
\caption{Voting schemes based on local marginal likelihoods. The local marginal likelihoods $\tilde{p}(y_s)$ are computed on all subsets $s$ for $6$ different models. In Figure \ref{fig:naive_local_bf} the local marginal likelihood of the correct model $6$ is compared with the other local marginal likelihoods. Based on a majority vote it is decided whether model $6$ wins. We indicate the win rate of model $6$, shown on the y-axis. In Figure \ref{fig:borda_count} we use a Borda voting scheme to aggregate the order of local marginal likelihoods into a global model ordering. We display the average Borda count and hence the correct model 6 must get the highest average count to win. 
As the number of splits increase (x-axis) the correct model 6 is not chosen anymore neither for the majority voting scheme nor for the Borda count. %
}

\label{fig:voting_schemes}
\end{center}
\vskip -0.3in
\end{figure}

\begin{figure}
    \begin{center}
    \begin{subfigure}[b]{0.49\textwidth}
        \includegraphics[width=0.8\columnwidth]{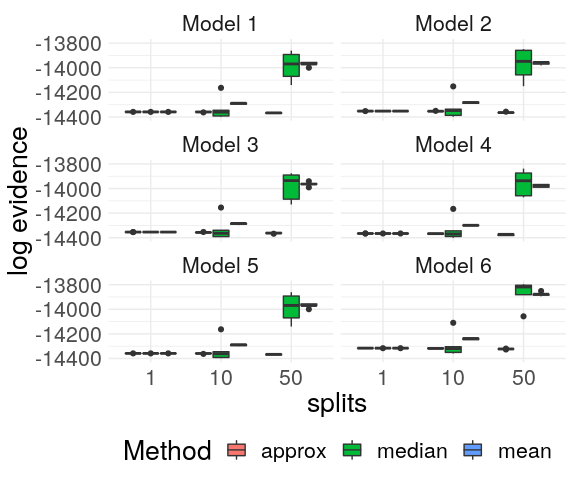}
        \caption{Comparison of three different methods for computing the marginal likelihood (y-axis). The boxes inside a number of splits (x-axis) are ordered as left our method (approx), middle the median and right the mean of the upsampled marginal likelihood.}
        \label{fig:linear_model_different_methods}
    \end{subfigure}
    \begin{subfigure}[b]{0.49\textwidth}
        \includegraphics[width=0.8\columnwidth]{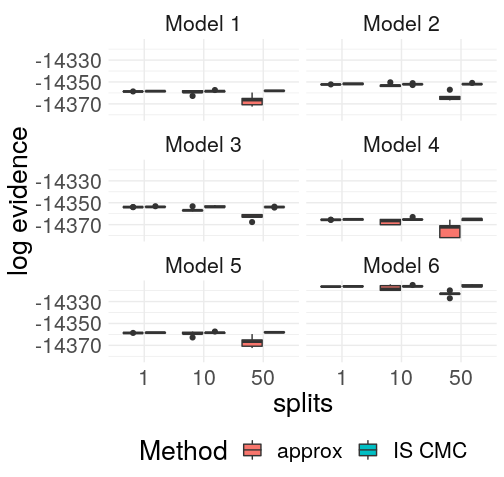}
        \caption{Comparison of the consensus Monte Carlo IS (CMC IS, right most boxes) approach for computing the marginal likelihood (y-axis) with our approximate method (approx, left most boxes) for a varying number of splits (x-axis).}
        \label{fig:linear_model_comparison_cmc}
    \end{subfigure}
       \caption{In Figure \ref{fig:linear_model_different_methods} (left) we compare the upsampled methods \{mean, median\} with our approximate method (approx) based on Algorithm \ref{alg:approxmethod}. The mean and median based approaches become unstable with 50 splits of the data. 
       In Figure \ref{fig:linear_model_comparison_cmc} (right) CMC IS stays competitive as the number of splits increase and seems more stable than our approximate method (approx) based on Algorithm \ref{alg:approxmethod}. However, an additional round of communication between the central node and the distributed workers is required. %
       }
       \label{fig:three graphs}
\end{center}
       \vskip -0.3in
\end{figure}

\subsection{Experiment 1: Voting schemes based on local marginal likelihoods and CMC IS} \label{ssec:voting_schemes}
In this section we illustrate potential issues of natural alternatives to our approach based on the decomposition in \eqref{eq:evidencedecomp}. Although this comparison is limited in scope, we believe it motivates and justifies the in-depth study of Algorithm \ref{alg:approxmethod}. 
This experiment is based on six different Gaussian regression models with a log normal prior on the variance that makes this model non-conjugate. We simulate a dataset with $10,000$ observations and induce high correlation of the features. The correct model (6) uses the all 17 covariates whereas the other models all omit one relevant variable (see also the Appendix for further details).

\paragraph{Majority voting and Borda counts}
To motivate the need for a coherent way of combining local evidences, we illustrate what can go wrong when using a naive approach for combining local inference. We try to identify the best model based on the local marginal likelihoods $\tilde{p}(y_s|m_k)$ for the models $m_k$.
We compare two different voting schemes that aggregate local rankings. First, we use a majority voting scheme, where each worker returns the winning model based on local marginal likelihoods, i.e., $\text{winner}_s = \arg \max_{m_k} \tilde{p}(y_s | m_k) \forall s$. Then, the central node aggregates the vote share of model 6 (i.e., on how many workers the correct model won): $\text{win rate} = \sum_{s=1}^S {1}_{\{ \text{winner}_s = 6 \} } / S $.
As a second voting scheme we use Borda counts \citep{bordacount}. The local models receive a score based on their ranking going from 1 to 6, i.e., $\text{Borda score}_{s, m_k} = \arg \text{sort}_{m_k} \tilde{p}(y_s | m_k)$, and the winning model is obtained by averaging the Borda scores $ \sum_{s=1}^S \text{Borda score}_{s, m_k} / S $. This voting scheme favors a consensus vote over a majority based decision. 

As illustrated in Figure \ref{fig:voting_schemes}, the voting schemes work reasonably well for a small number of splits. But with as little as $10$ splits (every worker sees $10\%$ of the data), the correct model is not chosen anymore, neither for the majority based voting nor for the Borda counts. Therefore, we decide not to compare our suggested method with these naive aggregations in the remainder of our experiments.  

\paragraph{Mean and median of upsampled normalising constants}
We use the same setting as above and compare two additional methods with our suggested approach from Algorithm \ref{alg:approxmethod}. 
We estimate locally the marginal likelihood of what we call an upsampled model following the idea in \cite{zhang2018robust}, i.e., we replicate locally the data of each shard to match the full data set size. The upsampled marginal likelihood is computed using the replicated data as
$$
\acute{p}(y_s) = \int_{\Theta}  p(y_s | \theta)^S p(\theta) d \theta.
$$
A final estimator of the marginal likelihood is obtained by computing the mean or median 
\begin{eqnarray*}
    \widehat{p}(y) = \begin{cases}
        \text{mean}(\{\acute{p}(y_1), \dots, \acute{p}(y_S) \}), \\
        \text{med}(\{\acute{p}(y_1), \dots, \acute{p}(y_S) \}),
    \end{cases}
\end{eqnarray*}
of the local marginal likelihoods when combining them on the central node. 
We illustrate the results of this approach in Figure \ref{fig:linear_model_different_methods}. Using medians or means of upsampled local marginal likelihoods becomes unstable when reaching $50$ shards compared to our approach in Algorithm \ref{alg:approxmethod}. We therefore do not consider this method in the rest of our experiments. 

\paragraph{Consensus Monte Carlo based importance sampling}
As a final point of comparison, we suggest a straightforward extension of CMC to the computation of the normalising constant based on importance sampling using the consensus normal distribution. This approach is based on the following importance sampling identity that can also be derived from Proposition \ref{prop:isub_normal}.
\begin{eqnarray}\label{eq:cmc_is}
 p(y) %
 = \int_{\Theta} \prod_{s=1}^S \frac{ p(y_s | \theta) p(\theta)^{1/S}}{ \mathcal{N}( \theta| \mu_s, \Sigma_s) } \prod_{s=1}^S \mathcal{N}( \theta| \mu_s, \Sigma_s)  d \theta \approx \frac{1}{N} \sum_{i=1}^N \prod_{s=1}^S \underbrace{\frac{ p(y_s | \theta_i) p(\theta_i)^{1/S}}{ \mathcal{N}( \theta_i| \mu_s, \Sigma_s) } }_{=: w_s(\theta_i)} \times \gamma,
\end{eqnarray}
where we sample $\theta_i \sim \mathcal{N}(\theta | {\mu}, {\Sigma}) \propto \prod_{s=1}^S \mathcal{N}( \theta_i| \mu_s, \Sigma_s)$ for $N$ samples and 
$$\gamma = \int_\Theta  \prod_{s=1}^S \mathcal{N}( \theta| \mu_s, \Sigma_s)d \theta,$$ 
correctly normalizes the importance sampling estimator in \eqref{eq:cmc_is}.  
Thus, we construct an approximation of the whole posterior but we only need to evaluate the ratio of the local subposteriors and the local normal approximation $w_s(\theta_i)$. This approach adds an additional round of communication with the central node, as $\theta_i \sim \mathcal{N}(\theta | {\mu}, {\Sigma})$ is generated on the central node, the weights $w_s(\theta_i)$ are computed locally and then send back to the central node to compute the aggregation. This IS approximation is exact, but can suffer from high variance if the local normal approximation $\mathcal{N}( \theta| \mu_s, \Sigma_s)$ and the consensus approximation $\mathcal{N}(\theta | {\mu}, {\Sigma})$ have little overlap. The product of importance weighting factors $ \prod_{s=1}^S w_s(\theta_i)$ potentially introduces further variance as the number of splits $S$ increase (although we did not observe this empirically in Figure \ref{fig:linear_model_comparison_cmc}). 

We illustrate the result of this method in Figure \ref{fig:linear_model_comparison_cmc}. Using an IS approximation based on CMC with 1000 importance samples performs reasonably well in this setting. A more detailed comparison of the CMC IS approximation for computing the marginal likelihood is out of scope for this work. 

\subsection{Experiment 2: Comparison of the conditional and approximate method} 
In this experiment we are interested in predicting on-time arrival of air planes where we use a Bayesian logistic regression model with $17$ 
features (see the appendix for more details and a comparison with another model). 
There are in total $n=327,346$ observations. Figure \ref{fig:flightsdata_m1} shows the results of our simulations. 
The reference value for the model evidence has been obtained using importance sampling using a Laplace approximation based on the \textit{maximum a posteriori}. %
In the current setting the bias is less than $-0.5\%$, even for as many as $50$ splits, both for the conditional approach that uses data augmentation (Algorithm \ref{alg:exactmethod}) and the approximate method based on normal approximations (Algorithm \ref{alg:approxmethod}).  
We notice a strong downward bias for the conditional approach, as predicted by Proposition \ref{thm:propisubvar}. 
This clearly illustrates a weakness of the conditionally conjugate approach and we observed this behaviour on different data sets (not shown here).

\begin{figure}
    \begin{center}
    \centerline{\includegraphics[width=0.6\columnwidth]{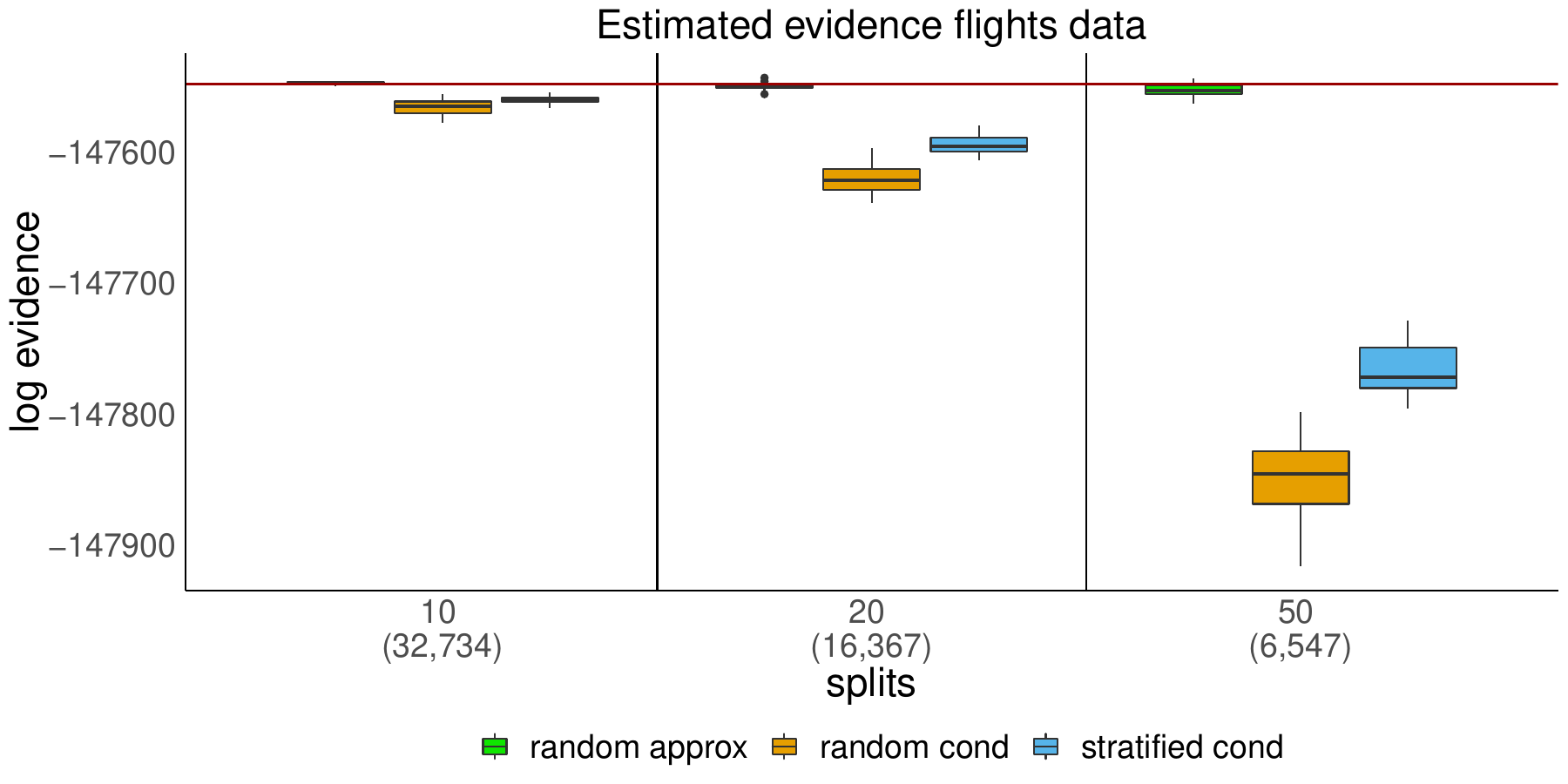}}
    \caption{Comparison of the calculated normalising constant (y-axis) for a logistic regression on the flights data. 
    As the number of splits increase (x-axis), the estimates become unreliable for the conditionally conjugate approach in Algorithm \ref{alg:exactmethod} (denoted by random cond, middle boxes), 
    even when using stratification (stratified cond, right boxes) for the sampling, whilst the approximate method (random approx, left boxes, Algorithm \ref{alg:approxmethod}) is stable and accurate. 
    The average number of observations per split is indicated in parentheses. The reference value is indicated as horizontal line. }
    \label{fig:flightsdata_m1}
    \end{center}
    \vskip -0.3in
    \end{figure}

A remedy for the high variance is the use of stratification to construct more homogeneous data shards to improve the performance of the conditional approach. 
We performed k-means clustering of the features with 10 clusters using the full data set. Then we stratify the observations using the outcome and the cluster membership. A similar approach was used in \citet{zhao2014accelerating} to diversify sampling for mini batches in stochastic gradient optimisation. The motivation is to obtain representative samples of the entire data set with every cluster being represented. 
Although often feasible in practice, this approach goes against the idea of distributed computation as all data has to be seen at once to construct a stratification. The improvement for the conditional approach that comes from stratification is rather limited, as shown in Figure \ref{fig:flightsdata_m1}. In the remaining experiments we consider only the approximate method described in Section \ref{ssec:approximateapproach} and no stratification. 

\subsection{Experiment 3: Assessment of the approximate method in a toy example} \label{ssec:toy_model}
In this experiment we investigate the behaviour of Algorithm \ref{alg:approxmethod} in the same Gaussian toy model as in the first experiment from Section \ref{ssec:voting_schemes}. We compare the performance for estimating the marginal likelihood over up to 
$50$ splits of $10,000$ observations. The error introduced via the approximation in \eqref{eq:gaussianisub} of the subposterior amounts to less 
than $0.02 \%$ (see Appendix), 
and despite a small downward bias in estimating $\log p(y)$, the resulting BFs stay stable over decreasing subset sizes (Figure \ref{fig:BF_gaussian}), 
resulting in a consistent choice of the correct model $(6)$. See also in the appendix for the 
comparison of log marginal likelihoods, that draws a similar picture. 

\begin{figure}
\begin{center}
\centerline{\includegraphics[width=0.8\columnwidth]{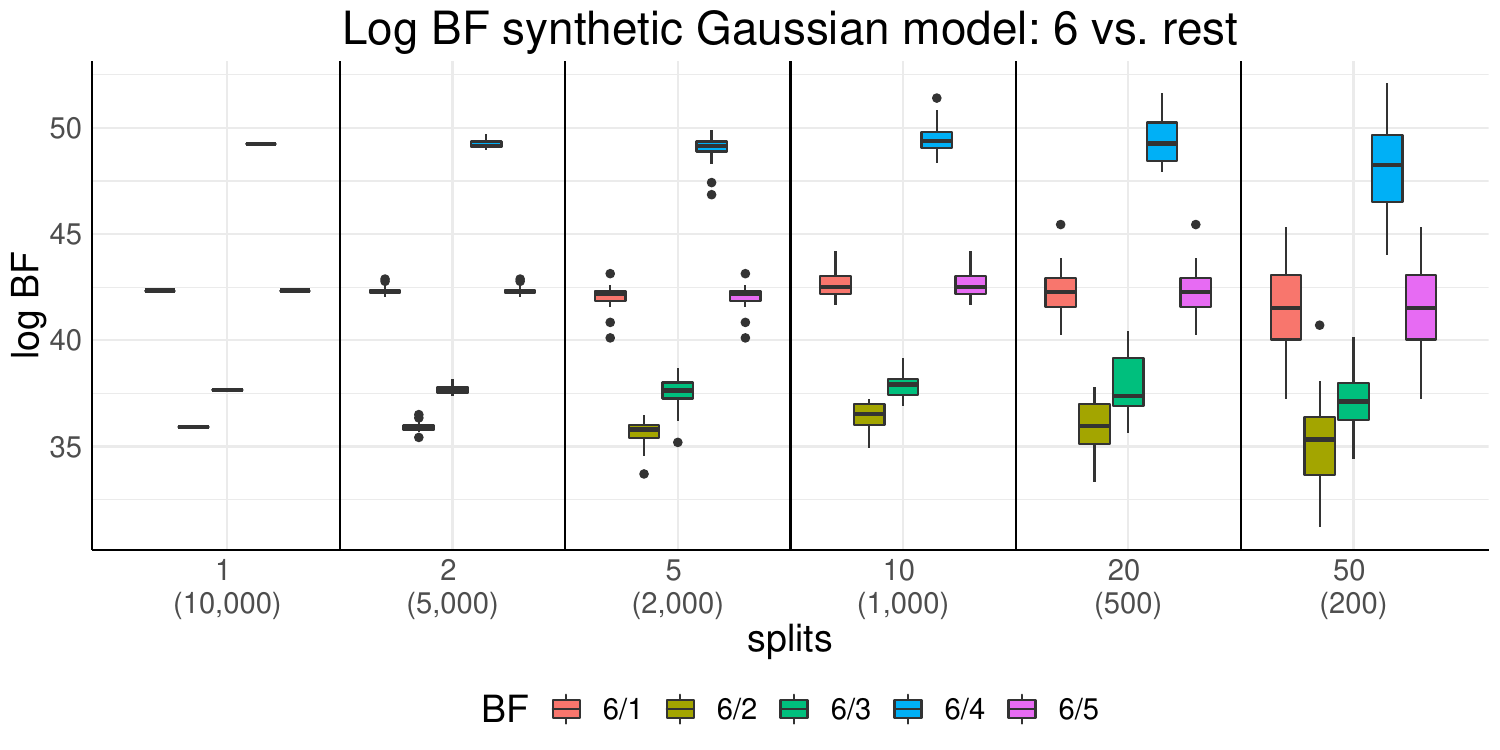}}
\caption{log Bayes factor computed using Algorithm \ref{alg:approxmethod} (y-axis) of the correct model $6$ against all other models. The BFs stay roughly constant even as the number of splits increase. The average number of observations per split is indicated in parentheses (x-axis). The value for a single split (left most column) serves as reference value. The ordering of the Bayes factors does not change as the number of splits increase and hence consistent model choice is possible over an increasing number of splits. }
\label{fig:BF_gaussian}
\end{center}
\vskip -0.3in
\end{figure}

\subsection{Experiment 4: The approximate method on a very large data set}
This experiment is based on a very large data set from particle physics where a binary classification problem consists in predicting 
the presence of a Higgs boson. The data set contains $11$ million observations. Our aim is to understand whether model (1) with 21 low-level 
features or model (2) with 7 high-level features is more likely \textit{a posteriori}. 
Running a full Monte Carlo simulation on the whole data set for model (1) leads to excessive computation times: 
a full run would take more than 450 hours (almost 3 weeks) on a single core CPU. We assess how the model evidence changes as splits get small 
by dividing the data set in shards of $1 \%, 0.2 \%$ and $0.1 \%$. 
Thereby we bring the computation time down to less than $5$ hours, $1$ hour and less than $30$ minutes, 
all running on different single core CPUs. The combination of the results of the different workers is in the 
order of a few seconds as we have to perform $\mathcal{O}(S p^3)$ operations to calculate $I_{\text{sub}}$. 
(Although necessary matrix inversion can be pre-computed locally before sending them to the central node.)
We show the results of our estimation in Figure \ref{fig:higgs}. Some bias is introduced by the splitting as 
the number of shards grow as illustrates the right hand side of Figure \ref{fig:higgs}. However, the bias is overall small and not visible when comparing both models on the same scale (left side of Figure \ref{fig:higgs}). 
Consequently, we would clearly choose model (1). 

In essence our experiment on the Higgs data set illustrates the necessity for distributed computation in the large data regime. 
Running the same experiment on the entire data set is too slow for most applications. 

\begin{figure}
\begin{center}
\centerline{\includegraphics[width=0.8\columnwidth]{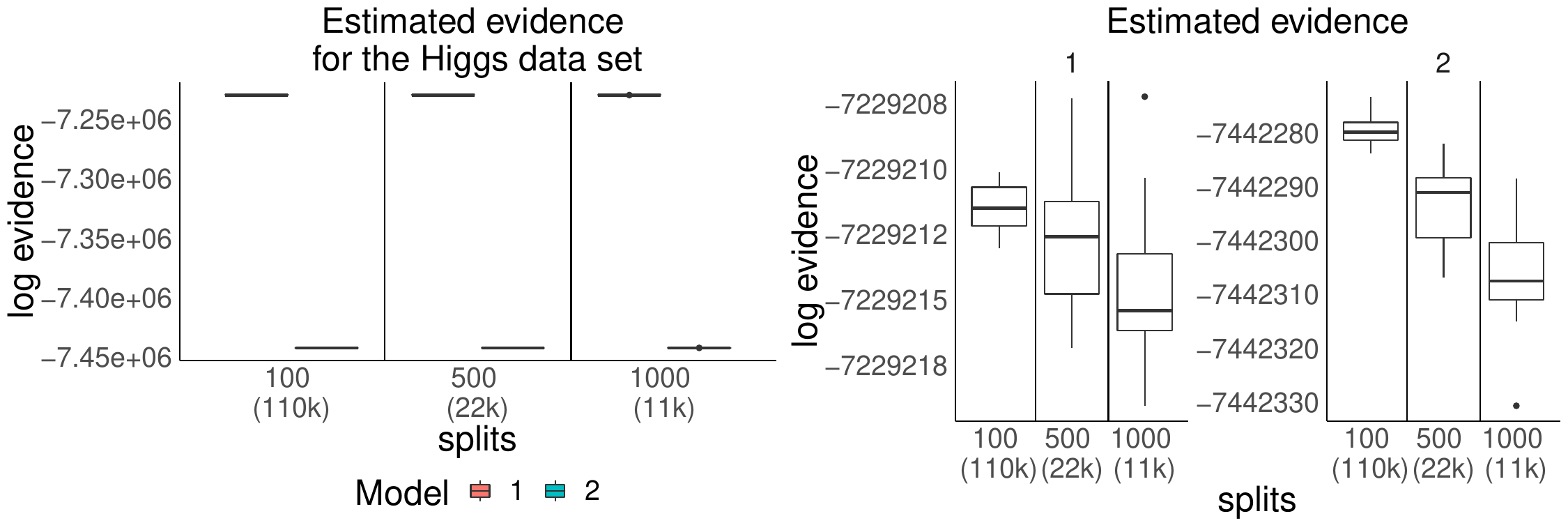}}
\caption{Comparison of the calculated normalising constant using Algorithm \ref{alg:approxmethod} (y-axis) for the Higgs data set. 
The left plot shows the evidence for both models (model 1 left boxes, model 2 right boxes) on the same scale. The middle plot corresponds to model (1), 
the right plot corresponds to model (2), using their respective scales only. The average number of observations per split is given in parentheses (x-axis). We observe a small downward bias as the number of splits increase.}
\label{fig:higgs}
\end{center}
\vskip -0.3in
\end{figure}

\subsection{Experiment 5: Sparse regression on a large genetic data set}
We compare the performance of a linear regression model with a Laplace prior on a real genetic data set from the UK
Biobank database. %
There are $n=132,353$ observations available. We consider model 1 with 50 and model 2 with 100 genetic variants in the human leukocyte antigen (HLA) region of chromosome 6 that are included as features in order to predict mean red cell volume (MCV). 
Due to the Laplace prior the conditionally conjugate approach is not applicable. %
We face again a situation where sampling the posterior on the entire data set on a single core CPU is estimated to take more than 200 hours for model 2. Using 20, 50 and 100 splits brings this computation time down to 10 hours, 2 hours and 1 hour. 

In order to assess the bias properly we decided to run our approximate method on a $10\%$ subset of the data where it is computationally feasible to analyse the whole subset and derive a reference value for the normalising constant. 
We see in Figure \ref{fig:hla} on the left side that a downward bias is present, but that we would choose consistently across subsets the right model for the $10\%$ subset of the full data set. 
When we use the full data set with $20$ to $100$ splits the right model is still chosen 
correctly in a given partition (see the right side of the same figure). We also see that model comparison becomes meaningless if different partitions are used to compare the model evidences.

\begin{figure}
    \vskip -0.1in
    \begin{center}
    \centerline{\includegraphics[width=0.8\columnwidth]{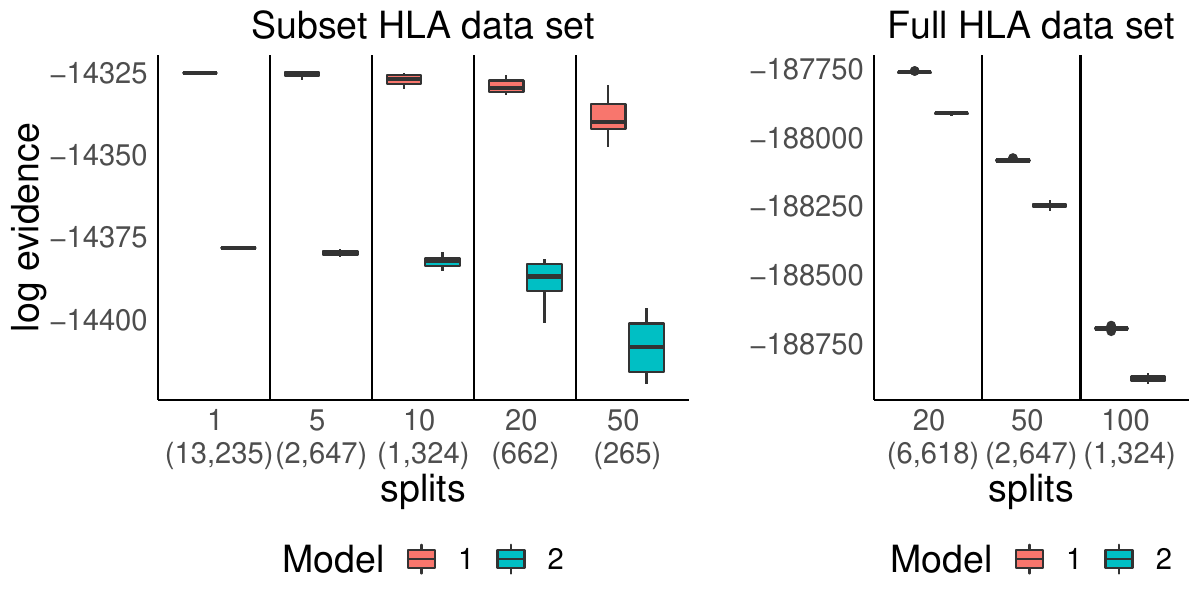}}
    \caption{Comparison of the calculated normalising constant using Algorithm \ref{alg:approxmethod} (y-axis) for a linear regression using a sparsity enforcing prior. Model 1 (left boxes) has 50 features, model 2 (right boxes) has 100 features. Left plot: estimated model evidence for the model run on a $10\%$ subset. Right plot: estimated model evidence using the full data set but starting with 20 splits. The average number of observations per split is given in parentheses (x-axis). We observe a downward bias as the number of splits increase. For the full data set a comparison across a different number of splits would be meaningless. }
    \label{fig:hla}
    \end{center}
    \vskip -0.4in
\end{figure}

As illustrates Table \ref{tab:hla}, the error relative to the true value of the normalising constant ($\% \sqrt{\text{MSE}}$) stays small even when using $50$ splits and thus having only $265$ observations for the estimation of $100$ parameters. As the number of splits increases, the squared bias starts to dominate the error as illustrates the ratio $\text{Bias}^2/\text{Var}$ in Table \ref{tab:hla}. %
The error relative to the level of the quantity that we are trying to estimate stays rather small.

\begin{table}[h]

\begin{center}
\begin{small}
\begin{sc}
\begin{tabular}{lrrrr} 
splits  & 5 & 10 & 20 & 50\\
\toprule
\multicolumn{5}{c}{Model 1 (50 features)} \\
$\sqrt{\text{MSE}}$ &    0.959 &  2.726 &  4.537 & 14.989\\
$\% \sqrt{\text{MSE}}$  & -0.007 & -0.019 & -0.032 & -0.105\\
$\frac{\text{Bias}^2}{\text{Var}}$  &  0.04 &  1.32 &  3.71 &  5.44\\
\midrule
\multicolumn{5}{c}{Model 2 (100 features)} \\
$\sqrt{\text{MSE}}$ & 1.812 &  4.591 & 11.404 & 30.996\\
$\% \sqrt{\text{MSE}}$  & -0.013 & -0.032 & -0.079 & -0.216\\
$\frac{\text{Bias}^2}{\text{Var}}$  &  3.42 &  6.07 &  2.42 & 11.99\\
\bottomrule
\end{tabular}
\end{sc}
\end{small}
\end{center}
\caption{Error metrics of the approximation of the log marginal likelihood for the subset $(10\%)$ of the HLA data set with $13,235$ observations.} \label{tab:hla}
\end{table}

Finally, for this experiment we assess the error of $\int_\Theta \prod_{s=1}^S \tilde{p}(\theta|y_s) d \theta$ as the number of splits grows, but the total data set size of $13,235$ stays fixed. 
We define the following two measures of the error. 

\begin{eqnarray*}
\epsilon_{1,S} = | \log \int_\Theta \prod_{s=1}^S \tilde{p}(\theta|y_s) d \theta - \log \int_\Theta \prod_{s=1}^S \mathcal{N}(\theta|\mu_s, \Sigma_s) d \theta |, \\
\epsilon_{2,S} = \log | \int_\Theta \prod_{s=1}^S \tilde{p}(\theta|y_s) d \theta - \int_\Theta \prod_{s=1}^S \mathcal{N}(\theta|\mu_s, \Sigma_s) d \theta | .
\end{eqnarray*}

The error $\epsilon_{1,S}$ (referred to as relative log error) corresponds to the log of the error predicted by Proposition \ref{prop:log_error} and it is expected to grow linearly in $S$. The error $\epsilon_{2,S}$ (referred to as log error of the difference) is a measure of the absolute error. The left most plot in Figure \ref{fig:error_hla_approx} illustrates that $\epsilon_{1,S}$ seems to grows linearly with $S$. The error $\epsilon_{2,S}$ seems to grow more like $\log S$, 
as depicted in the right-most plot in Figure \ref{fig:error_hla_approx}. 

\begin{figure}
\begin{center}
\centerline{\includegraphics[width=0.7\columnwidth]{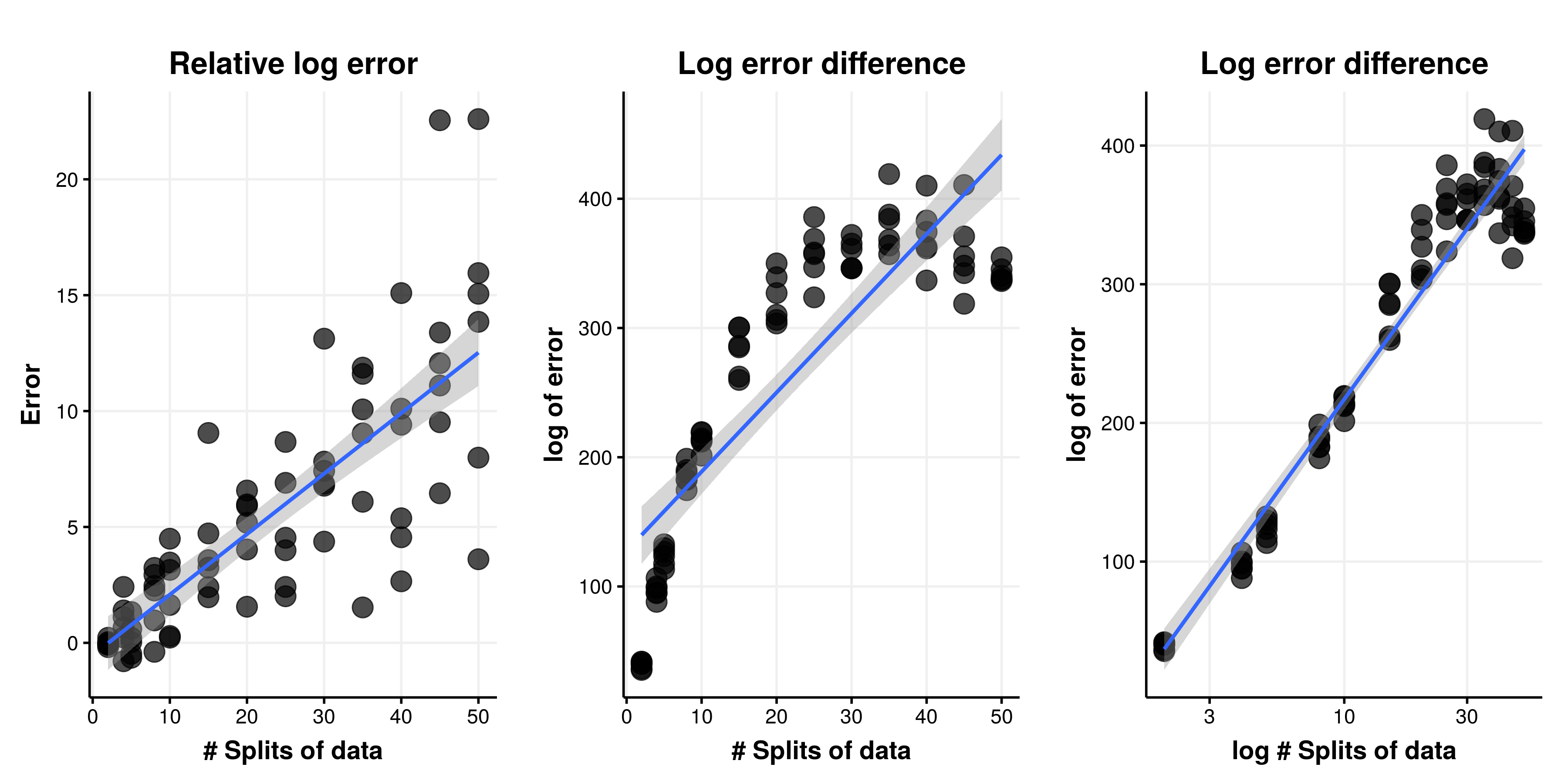}}
\caption{Left plot: error $\epsilon_{1,S}$ as a function of $S$. Middle plot: error $\epsilon_{2,S}$ as a function of $S$. 
Right plot: error $\epsilon_{2,S}$ as a function of $\log S$. 
The line corresponds to a linear regression line fitted to the data in order to illustrate the trend. 
We run every simulation $6$ times in order to get repeated measurements. 
The sampler is run over the range of splits $[1,2,4,5,8,10,15,20,25,30,35,40,45,50]$. The total number of observations is fixed to $13,235$ and 
varies accordingly with the number of splits.}
\label{fig:error_hla_approx}
\end{center}
\vskip -0.3in
\end{figure}

\subsection{Experiment 6: Distributed RJMCMC}
Finally, we investigate the use of our splitting approach for a vanilla model selection in a RJMCMC setting. Our simulation should be seen as a proof of concept as RJMCMC faces numerous issues that make exact inference dependent on various tuning parameters that go beyond the scope of this paper.

For this purpose we simulate a toy data set of size $n=4,000$ with a binary outcome where the five features exhibit a high correlation of $0.9$. 
\begin{table}[h!]
\begin{center}
\begin{footnotesize}
\begin{sc}
\begin{tabular}{lrrrrrr} 
 Variables & 1 & 2 & 3 & 4 & 5 & \\
  \toprule
 Model 1 & \checkmark  & \checkmark & \checkmark & $\times$ & \checkmark & \\
 Model 2 & \checkmark  & $\times$ & $\times$ & \checkmark & \checkmark & \\
 Model 3 & \checkmark  & \checkmark & $\times$ & $\times$ & \checkmark & \\
\midrule
 Truth & \checkmark  & \checkmark & $1/2$ & $\times$ & \checkmark & \\
\bottomrule
\end{tabular}
\end{sc}
\end{footnotesize}
\end{center}
\caption{Active variables in the RJMCMC experiment.} \label{tab:rjmcmc}
\end{table}

The data is generated by mixing two data sets. %
Thus, we artificially generate a setting where it is not clear whether to include the third variable. See Table \ref{tab:rjmcmc} for more details as well as the appendix. 
The comparison of the different models is shown in Figure \ref{fig:rjcmcm}. As the number of splits increase, the estimates of the Bayes factors deteriorate. 
In the current setting going beyond $3$ splits may lead to misleading results due to high variance of the estimates as show our experiment. 
This high variance occurs when combining back the results and is due to several reasons. RJMCMC samplers take a long time to mix and less likely models are potentially not explored enough. Therefore both the estimates based on the MCMC samples of the chain as well as the sojourn times suffer from high variance that accumulates when combining the results from several splits. 
There is no guarantee that on all data shards all models of interest are explored, if the data shards are too small. In this situation it is not possible to combine the results from several shards back together. 

For our experiments we decided to use a medium sized data set and only a small number of features and thereby making exploration of the relevant models more likely. %
Potential remedies for the evoked problems are an improved estimation of the BF using the method presented in \citet{bartolucci2006efficient} or the construction of more homogeneous splits using stratification. In any case, whilst theoretically possible, distributed model evidence will rely on well mixing efficient RJMCMC in each split.

\begin{figure}
    \begin{center}
    \centerline{\includegraphics[width=.7\columnwidth]{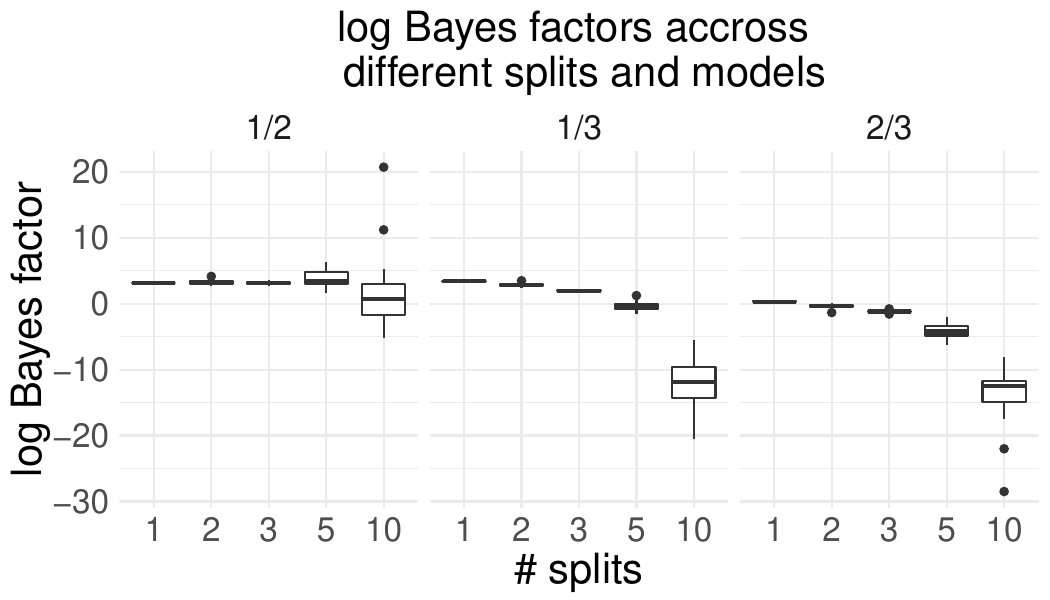}}
    \caption{Comparison of the Bayes factor over several splits for the RJMCMC sampler based on Algorithm \ref{alg:rjmcmcmethod} (y-axis). We compare the BF of model 1 vs 2 (1/2), reference log BF $3.2$, model 1 vs 3 (1/3), reference log BF $3.6$ and model 2 vs 3 (2/3), reference log BF $0.4$ over a changing number of splits (x-axis).}
    \label{fig:rjcmcm}
    \end{center}
    \vskip -0.3in
\end{figure}

\subsection{Guidelines for practical application}
When using our approach in practice we recommend at least a few thousand observations per data shard for a normal approximation to be reasonable. At this stage we suggest to avoid high dimensional settings where the number of parameters exceeds the number of data points, see also \citet{kass1995bayes}, where at least 5 observations per dimension are recommended. If the number of observations per shard are too small, the normal approximation becomes unreliable and one risks to face a large downward bias when combining the results. For complex models the bias grows faster %
as we split the data in smaller shards. 
It can be helpful to evaluate the variance of summary statistics across the shards to detect if the splits are not homogeneous and run the sampler on a different number of splits. 

We recommend sufficiently long Markov chains (see our default settings for the experiments) and the use of convergence diagnostics to make sure that the posterior has been explored sufficiently and that estimated posterior moments are reliable. 
In particular, we often face the challenge to find a balance between 
(a) making sure the approximations are precise enough and (b) limiting computation time. 
In practice, the bias in the estimation is often smaller than the variation in the estimators. %
Thus consistent model choice between competing models is possible. However, care is needed if competing models are similar.%

\section{Discussion and Conclusion} \label{sec:discussion} 
We have presented an approach to calculate the normalising constant in a distributed fashion to enable Bayesian model choice with large data sets. 
We are able to effectively divide the computation time by several orders of magnitude by splitting the data over a large number of workers and limiting communication between workers. 
We have shown overall good numerical results and explained the theoretical underpinning for our approach. There remain open questions. 
Although the estimation of the subposterior normalising constants is biased in general, 
this bias seems worth accepting in practice.  
It would be interesting to link this bias of $I_{\text{sub}}$ to the way the data is split and to the characteristics of the model such as, e.g., its dimension. %

Proposition \ref{prop:isub_normal} suggests possible refinements of the normal approximation to $I_{\text{sub}}$ \eqref{eq:gaussianisub}. In principle, the correction factor 
\begin{align}
  \dfrac{\int_{\Theta} \prod_{s=1}^{S}\widetilde{p}({\theta}\mid {y}_{s}) \ d\theta}{ {{\int_{\Theta}  \prod_{s=1}^{S}\mathcal{N}({\theta}\mid {{\mu}}_{s}, {{\Sigma}}_{s}) \ d{\theta}}}} &=  {{\mathbb{E}_{\mathcal{N}({\theta}| {{\mu}}, {{\Sigma}})}\left[ \prod_{s=1}^{S}\dfrac{\widetilde{p}({\theta} \mid {y}_{s})}{\mathcal{N}({\theta}| {{\mu}}_{s}, {{\Sigma}}_{s})} \right]}} \label{eq:correction}
\end{align}
could be estimated by first approximating the density ratio terms $\widetilde{p}(\theta | y_{s})/\mathcal{N}(\theta | \mu_{s}, \Sigma_{s})$, and then using a simple Monte Carlo average for the expectation over the global normal approximation $\mathcal{N}(\theta | \mu, \Sigma) \propto \prod_{s=1}^{S}\mathcal{N}(\theta | \mu_{s}, \Sigma_{s})$ as we suggest for the IS CMC estimator. Subposterior samples could be used to construct a kernel density estimate of $\widetilde{p}(\theta | y_{s})$ \citep{neiswanger2013asymptotically}. The density ratio could also be estimated using various nonparametric techniques from the machine-learning literature \citep{kanamori_2012_statistical}. For large values of $S$, a Laplace approximation to the expectation could also be accurate enough for practical purposes \citep{tierny_1986_accuracte}.  Estimation of the correction term \eqref{eq:correction} will introduce some additional variance, however this may be compensated for by a reduction in the bias. A detailed investigation of the costs and benefits of these approaches is an avenue for future research. 

If we want to achieve truly parallel Bayesian computation, we must be able to both split the data and run short Markov chains without burn-in bias \citep{jacob2017unbiased}. Based on this idea an unbiased estimation of the normalising constant via bridge sampling \citep{rischard2018unbiased} could be combined with our method to improve scalability. %
Another interesting avenue for future research would be the use of variational inference for distributed Bayesian model choice using the work of \citet{rabinovich2015variational,nowozin2018debiasing}. 
In practical settings shotgun stochastic search (SSS) \citep{hans2007shotgun} could be applicable using our decomposition as SSS relies on normal approximations to the posterior to quickly explore different models. 
We also suggest to generalise our approach to settings where data subsets are not i.i.d. and of varying seize and investigate 
applications to federated learning \citep{li2020federatedlearning} in combination with posterior approximations for neural networks \citep{immer2021scalable}. 
We think that distributed Bayesian computation merits further theoretical and practical investigation as an alternative to the mini batch paradigm.  

\subsubsection*{Acknowledgements}
We thank Leonardo Bottolo, Paul Newcombe, Will Astle and Nicolas Chopin for helpful discussions and Will Astle for providing data. 
We would like to thank the reviewers and the editor for their feedback that greatly helped to improve our work. This work was supported by the EPSRC (EP/R018561/1), an MRC programme grant (MC\_UU\_00002/10) 
and the Alan Turing Institute (TU/B/00092).

\bibliographystyle{ba}
\bibliography{example_paper}

\appendix
\newpage
\clearpage
\section{Additional details on the algorithm}
We present in the following additional details on the algorithms. First we show how to use P\'{o}lya-Gamma data augmentation for the logistic regression 
in a distributed setting. Then we focus on the RJMCMC approach.

\subsection{P\'{o}lya Gamma data augmentation for the logistic regression} \label{appendix:polyagamma}
We observe a vector of binary outcomes $y \in \{0,1\}^n$ depending on some feature matrix $X \in \mathbf{R}^{n\times p}$. We assume that $y_i \sim \mathcal{B}(p_i)$, where $\mathcal{B}(\cdot)$ denotes a Bernoulli distribution and $p_i$ is the probability of observing $y_i = 1$. $p_i = \text{logit}^{-1}(x_i^t \theta)$, where $\text{logit}$ is the logit transform and $\theta \in \mathbf{R}^p$ is the unknown parameter vector endowed with a multivariate Gaussian prior: $\theta \sim \mathcal{N}(m_0,V_0)$. In our applications we will set $m_0 = 0$.
We recall the two distributions we sample from in a Gibbs sampler (now for the subset of the data $y_s$):
\begin{eqnarray}
\tilde{p}(\theta |z_s, y_s) = \mathcal{N}(m_s,V_s) \\
\tilde{p}(z_{s,(j)} | \theta, y_{s,(j)}) = \PG(1, x_{s,(j)}\theta), 
\end{eqnarray}
where $x_{s,(j)}$ denotes the $j$th observation in shard $s$ and 
\begin{eqnarray}
V_s = \left( X^t_{s}\Omega X_{s} + s^{-1} V_0^{-1} \right)^{-1}, \\
m_s = V \left( X^t_{s} \kappa + + s^{-1} V_0^{-1}m_0 \right).
\end{eqnarray}
Here $\kappa = y - 1/2$ and $\Omega = \diag z_s$. 
The P\'{o}lya-Gamma distribution can be characterised as an infite sum of Gamma distributed random variables. In particular, $z \sim \PG(c,b)$ for $b>0$ and $c \in \mathbb{R}$ if $$z = \frac{1}{2 \pi^2} \sum_{k=1}^\infty \frac{g_k}{(k-1/2)^2+(c/(2 \pi))^2},$$ and $g_k \sim \Gamma(b, 1)~\text{i.i.d.}~\forall k$.
See \citet{polson2013bayesian} for more details on how to sample from the P\'{o}lya-Gamma distribution.

\subsection{Distributed RJMCMC} \label{appendix:rjmcmc}
In reversible jump Markov chain Monte Carlo one typically uses a prior on the model $p(m_k)$ that reflects the prior belief on the complexity of the model. A common choice is the beta-binomial prior, that has an inherent multiplicity correction \citep{wilson2010bayesian}. The posterior probability of model $p(m_k|y)$ is obtained by the relative time the samplers spends in this model. Using the identity 
$$
\frac{p(m_k|y)}{p(m_{k'}|y)} = \frac{p(y|m_k)}{p(y|m_{k'})} \frac{p(m_{k'})}{p(m)},
$$
we obtain the Bayes factor $\frac{p(y|m_k)}{p(y|m_{k'})}$ by correcting for the prior odds. %

Interestingly, the idea of splitting the data and running a sampler on the shards is also applicable in this setting. 
By using the decomposition in \eqref{eq:evidencedecomp} we obtain 
\begin{eqnarray} \label{eq:rjmcmc}
 \frac{p(y|m_1)}{p(y|m_2)}  = \frac{\prod_{s=1}^S \tilde{p}(y_s|m_1)   \alpha_1^S  \int_\Theta \prod_{s=1}^S \tilde{p}(\theta|y_s, m_1) d \theta }{\prod_{s=1}^S  \tilde{p}(y_s|m_2)   \alpha_{2}^S  \int_\Theta \prod_{s=1}^S \tilde{p}(\theta|y_s, m_2) d\theta }  \nonumber \\ 
= \left\{ \prod_{s=1}^S  \frac{\tilde{p}(m_1|y_s) p(m_2)}{\tilde{p}(m_2|y_s) p(m_1)} \right\} \frac{\alpha_1^S  \int_\Theta \prod_{s=1}^S \tilde{p}(\theta|y_s, m_1) d \theta }{\alpha_{2}^S \int_\Theta \prod_{s=1}^S \tilde{p}(\theta|y_s, m_2) d\theta }, 
\end{eqnarray}
where $\tilde{p}(m_1|y_s)/\tilde{p}(m_2|y_s)$ is obtained as the posterior odds ratio of the RJMCMC sampler of data subset $s$. 

The quantity $\tilde{p}(m_i|y)$ is typically available as the output of the sampler. The normalising constant of the subprior $\alpha_i$ for model $i$ is available for common priors. We can again use a normal approximation to $\tilde{p}(\theta|y_s, m_i)$ using the samples generated from the sampler. 
Consequently, the splitting approach can effectively be combined with a reversible jump algorithm as we suggest in Algorithm \ref{alg:rjmcmcmethod}.

\section{Additional details on the models and simulations}

\subsection{Logistic regression models}

\paragraph{The flights data}
This data set is available through the \texttt{nycflights13}\footnote{https://github.com/hadley/nycflights13} \texttt{R} package. This data set contains airline on-time data for all flights departing NYC in 2013. We create a binary indicator for the arrival delay if the flight arrived at least $1$ minute late. We use as explanatory variable the departure delay from the departing airport and to which airline the plane belongs to. After removing missing values we get a data set of $n=327,346$ observations. We consider two different models to explain the outcome $y$: (1) a model with a dummy variable per carrier and the departing delay in minutes. This yields $17$ different features. 
(2) a model with interactions between carrier and departing delay in addition to the other features. This yields a model with $32$ different features. 

\begin{figure}[ht]
\vskip -0.1in
\begin{center}
\centerline{\includegraphics[width=0.99\columnwidth]{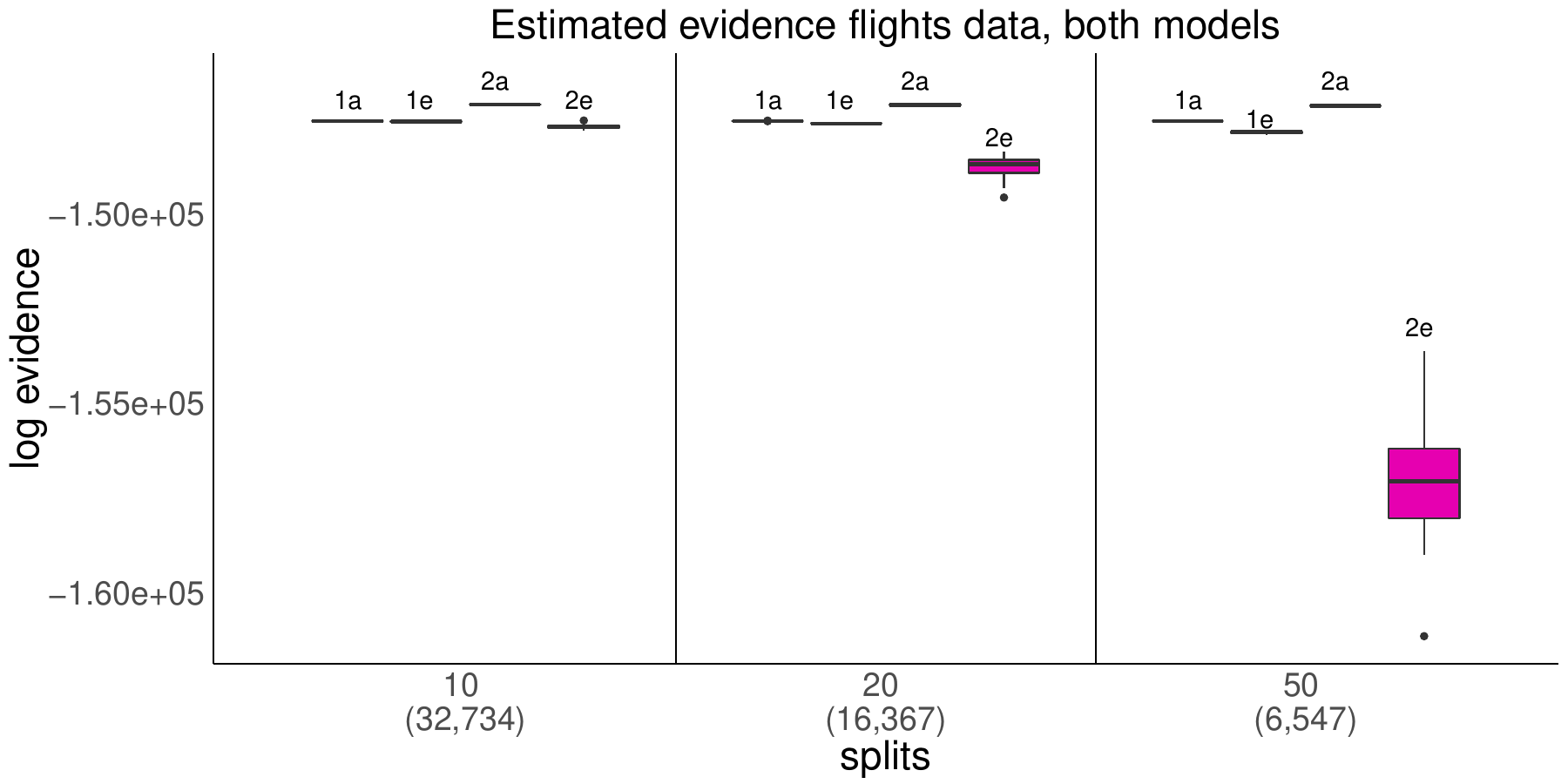}}
\caption{Comparison of the calculated normalising constant for a logistic regression on the flights data. As the number of splits increase, the estimates become unreliable for the conditionally conjugate approach. The left plot compares the conditional and approximate estimation. "1/2" stand for model 1/2 and "a/e" stand for the approximate (a) or conditional method (e). The average number of observations per split is indicated in parentheses. }
\label{fig:flightsdata_comp}
\end{center}
\vskip -0.2in
\end{figure}

For the approximate method one would clearly favour model (2) over model (1), as the evidence is higher.
This is maintained consistently as the number of splits increases despite a small downward bias, %
see Figure \ref{fig:flightsdata_comp}. 
The conditionally conjugate sampler starts breaking down for the more complex model (2) due to the variance of the estimator of $I_{\text{sub}}$. The estimated value of the model evidence for model (2) is below the value of model (1) for as little as 10 splits. Therefore a consistent model choice is not possible with the conditional approach in this case. 

\paragraph{The Higgs data}
Another binary regression experiment is based on a data set from particle physics where the classification problem consists in distinguishing between a signal process which produces Higgs bosons and a background process which does not. It is available through the UCI repository\footnote{https://archive.ics.uci.edu/ml/datasets/HIGGS}. The data set contains $11$ million observations and in total $28$ features.  The first 21 features are kinematic properties measured by the particle detectors in the accelerator. The last 7 features are high level features. Our aim in this task is to understand whether model (1) with the kinematic features or model (2) with the high level features is more likely \textit{a posteriori}. 

\subsection{Gaussian toy model} \label{appendix:gaussiantoymodel}
We use a synthetic Gaussian setting (linear model with 17 covariates, a log Normal prior on the variance term leading to a non-conjugate model). The true model $(\#6)$ contains all features, the other models omit one relevant variable each. High correlation of the features $(0.9)$ makes model choice difficult. The true parameters of the model are simulated from a normal distribution with varying mean and variance. The exact details are availabe in \url{https://github.com/alexanderbuchholz/distbayesianmc/blob/master/R/f_load_data.R#L61}. 
Figure \ref{fig:approx_error_gaussian} illustrates that the contribution of the product of normals (left plot) is rather low and the error is reasonably small (right plot). 
Figure \ref{fig:log_evidence_gaussian} shows again the presence of a downward bias when the number of splits increases, but this bias does not affect the ordering of the estimated marginal likelihoods. 

\begin{figure}[ht]
\begin{center}
\centerline{\includegraphics[width=0.99\columnwidth]{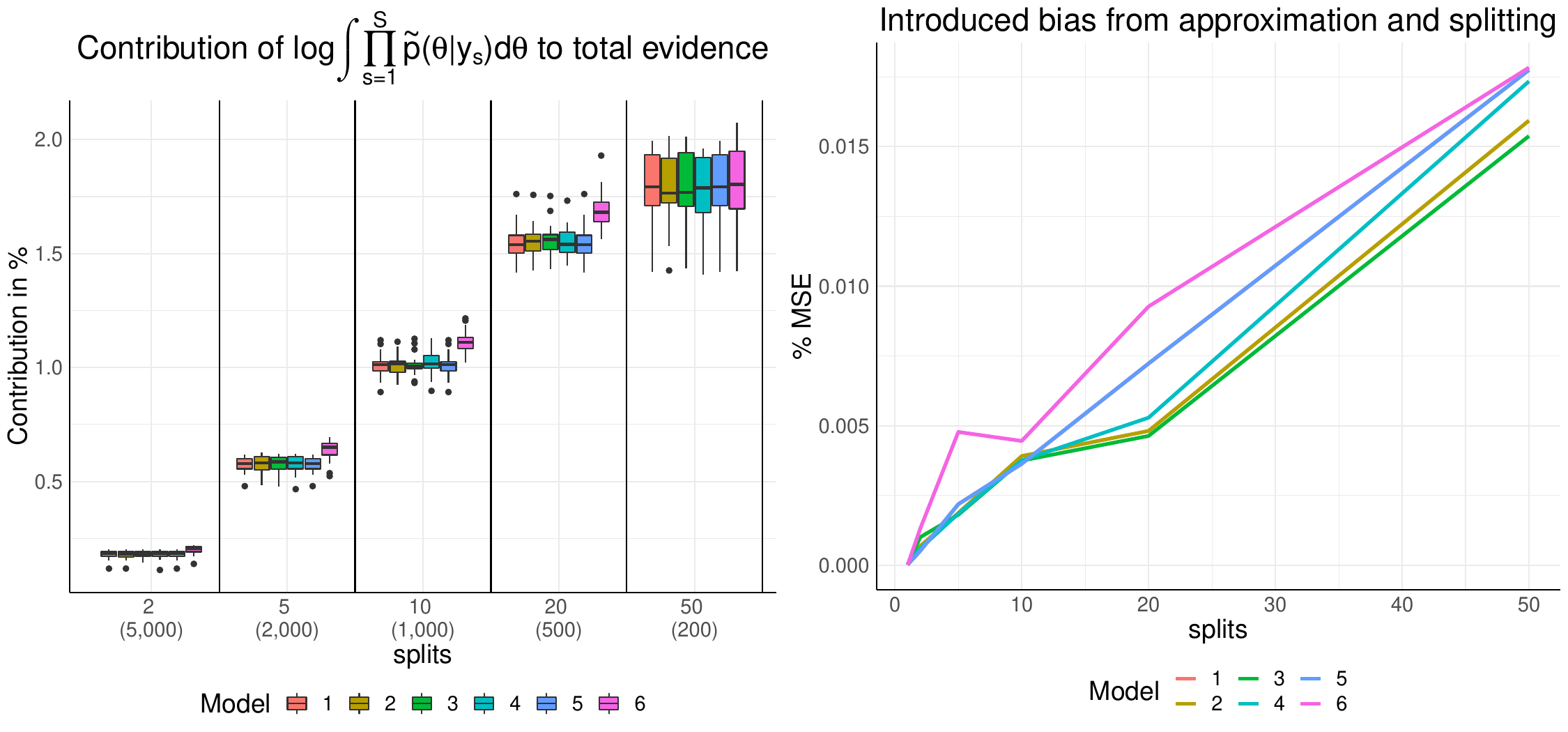}}
\caption{Left plot: contribution of $\log \int_\Theta \prod_{s=1}^S \tilde{p}(\theta|y_s) d \theta$ to the estimation of $\log p(y)$ in $\%$ over an increasing number of splits. Right plot: relative MSE as a function of the number of splits. Overall the contribution of $\log \int_\Theta \prod_{s=1}^S \tilde{p}(\theta|y_s) d \theta$ and error from the entire approximation stay small.}
\label{fig:approx_error_gaussian}
\end{center}
\vskip -0.2in
\end{figure}

\begin{figure}[ht]
\begin{center}
\centerline{\includegraphics[width=0.99\columnwidth]{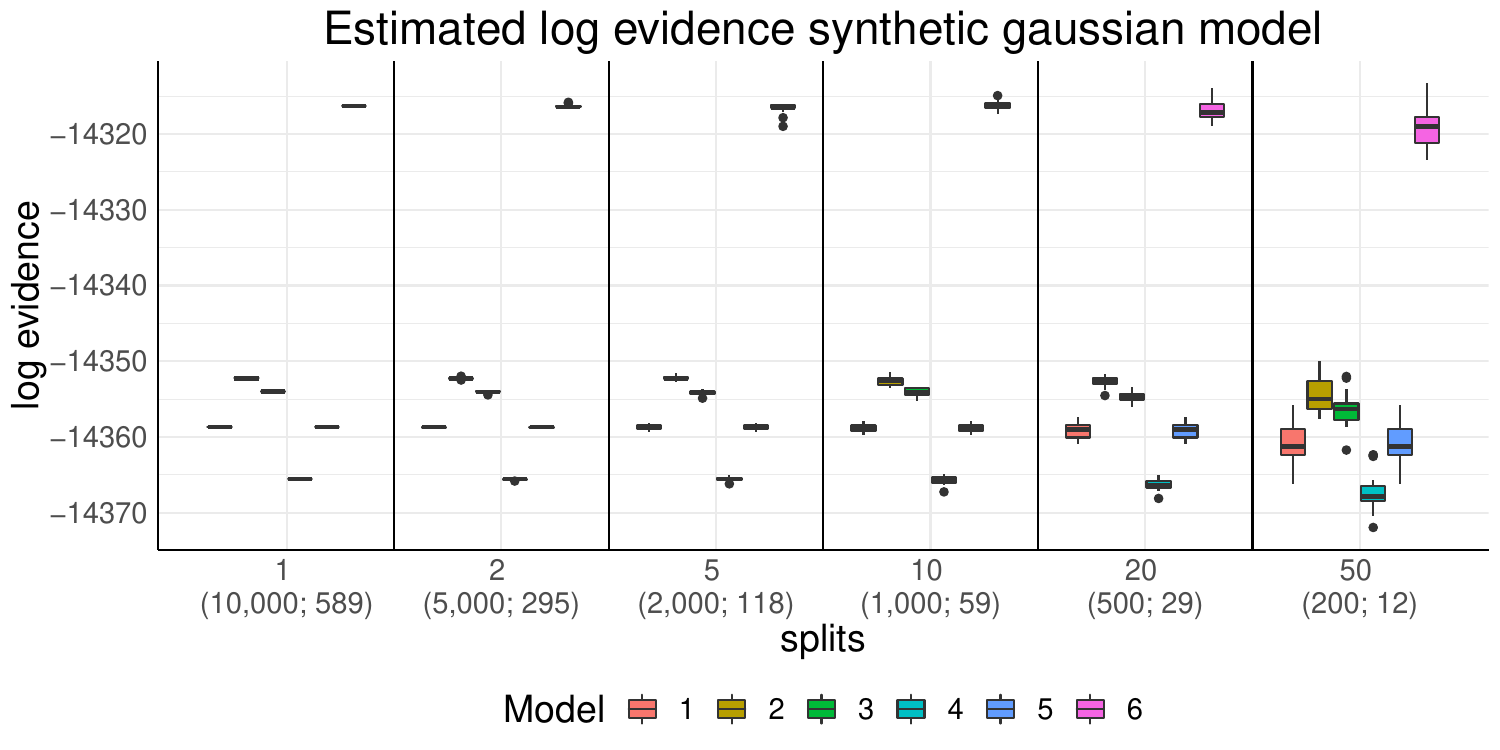}}
\caption{log marginal likelihood of the $6$ different models over an increasing number of splits. A slight downward bias is present. However, the bias does not change the ordering of the log marginal likelihoods.}
\label{fig:log_evidence_gaussian}
\end{center}
\vskip -0.2in
\end{figure}

\subsection{Sparse linear regression model}
We observe a vector of continuous outcomes $y \in \mathbf{R}^n$ depending on some feature matrix $X \in \mathbf{R}^{n\times p}$. We assume that $y_i \sim \mathcal{N}(\mu_i, \sigma^2)$, where $\mathcal{N}(\cdot)$ denotes a Gaussian distribution with mean $\mu_i$ and variance $\sigma^2$. $\mu_i = x_i^t \theta$, where $\theta \in \mathbf{R}^p$ is the unknown parameter vector endowed with a prior: $\theta \sim \mathcal{L}(0_p, \sigma_0 I_p)$, where $\mathcal{L}(\cdot, \cdot)$ is a Laplace (double exponential) prior. 

\paragraph{The HLA data set}
For our fourth experiment we compare the performance of a linear regression model with a Laplace prior on a real genetic data set from the UK Biobank database.
The selected outcome variable is mean red cell volume (MCV), taken from the full blood count assay and adjusted for various
technical and environmental covariates. Genome-wide imputed genotype data in expected allele dose
format are available on  $n=132,353$ study subjects. We consider 50 and 100 genetic
variants in the human leukocyte antigen (HLA) region of chromosome 6, selected so that the allelic scores has the highest absolute correlation with the outcome. The region was chosen as many
associations were discovered in a genome-wide scan using univariate models \citep{astle2016allelic}.

\subsection{Logistic regression model in the RJMCMC setting}
We observe a vector of binary outcomes $y \in \{0,1\}^n$ depending on some feature matrix $X \in \mathbf{R}^{n\times p}$. We assume that $y_i \sim \mathcal{B}(p_i)$, where $\mathcal{B}(\cdot)$ denotes a Bernoulli distribution and $p_i$ is the probability of observing $y_i = 1$. $p_i = \text{logit}^{-1}(x_i^t \theta)$, where $\text{logit}$ is the logit transform and $\theta \in \mathbf{R}^p$ is the unknown parameter vector endowed with a multivariate Gaussian prior: $\theta \sim \mathcal{N}(0_p, \sigma^2 I_p)$.
We generate simulated data as following:
We generate a highly correlated feature matrix $X$. Then we split the generated features in two and calculate $\mu_1 = X_{1} \theta_1$ and $\mu_2 = X_{2} \theta_2$, where $\theta_1 = [-1,1,0,0,1]$ and $\theta_2 = [-1,1,0.01,0,1]$. Thereby we have effectively half of the observations that will be better explained by including the third feature. 

We run the RJMCMC sampler for 10 million iterations where the first 2 million iterations are discarded as burn-in. We check for proper mixing by examining the trace plots of the sampler.

\section{Proofs}
\subsection{Proof of Proposition \ref{thm:prop1}}
\begin{proof}
We rewrite the the posterior as 

\begin{eqnarray*}
    p(\theta|y) &=&  \frac{p(y|\theta) p(\theta)}{p(y)} = \frac{\prod_{s=1}^S p(y_s|\theta) p(\theta)^{1/s}}{p(y)}  \\
  &=& \alpha^S \frac{\prod_{s=1}^S p(y_s|\theta) \tilde{p}(\theta)}{p(y)} ,
\end{eqnarray*}
where we have used the fact that $\tilde{p}(\theta) = p(\theta)^{1/S}/\alpha$. 
Now note that $p(y_s|\theta) \tilde{p}(\theta) = \tilde{p}(\theta|y_s) \tilde{p}(y_s)$, where all the distributions are correctly normalized, as indicated by the tilde. Plugging this decomposition in the previous equation we get
\begin{eqnarray} \label{eq:prodalphaposterior}
    p(\theta|y) 
  = \alpha^S \frac{\prod_{s=1}^S  \tilde{p}(\theta|y_s) \tilde{p}(y_s)}{p(y)}.
\end{eqnarray}
This is rewritten as 
\begin{eqnarray*}
    p(\theta|y) p(y)
  = \alpha^S \prod_{s=1}^S  \tilde{p}(\theta|y_s) \tilde{p}(y_s).
\end{eqnarray*}
And after integrating over $\theta$ we get
\begin{eqnarray*}
    \int p(\theta|y) d \theta p(y)
  = \alpha^S \int \prod_{s=1}^S  \tilde{p}(\theta|y_s) \tilde{p}(y_s) d \theta.
\end{eqnarray*}
As $p(\theta|y)$ integrates to one we obtain 
\begin{eqnarray*}
     p(y)
   = \alpha^S \prod_{s=1}^S  \tilde{p}(y_s) \int \prod_{s=1}^S  \tilde{p}(\theta|y_s) d \theta.
\end{eqnarray*}
\end{proof}

\subsection{Proof of Proposition \ref{thm:propisub}}
\begin{proof}
We recall the definition of $I_{\text{sub}}$:
$$
I_{\text{sub}} = \int_\Theta \prod_{s=1}^S \tilde{p}(\theta|y_s) d \theta.
$$
Let us now rewrite $\prod_{s=1}^S \tilde{p}(\theta|y_s)$ making use of the latent variables $z_s$ with latent subposterior $\tilde{p}(z_s|y_s)$:
\begin{eqnarray*}
&& \prod_{s=1}^S \tilde{p}(\theta|y_s) \\
&=& \prod_{s=1}^S \int_{\mathcal{Z}_s} \tilde{p}(\theta|y_s, z_s) \tilde{p}(z_s|y_s) d z_s \\
&=&  \int_{\mathcal{Z}_{1:S}}  \prod_{s=1}^S \left( \tilde{p}(\theta|y_s, z_s) \tilde{p}(z_s|y_s) \right) d z_1, \cdots d z_S \\
&=& \int_{\mathcal{Z}_{1:S}}   \left( \prod_{s=1}^S \tilde{p}(\theta|y_s, z_s) \right) \left( \prod_{s=1}^S \tilde{p}(z_s|y_s) \right) d z_1, \cdots d z_S
\end{eqnarray*}
When integrating out $\theta$ and plugging in the previous equation we obtain
\begin{eqnarray*}
&& \int_{\Theta} \prod_{s=1}^S \tilde{p}(\theta|y_s) d \theta \\
&=& \int_{\Theta} \int_{\mathcal{Z}_{1:S}}  \left( \prod_{s=1}^S \tilde{p}(\theta|y_s, z_s) \right) \left( \prod_{s=1}^S \tilde{p}(z_s|y_s) \right) d z_1, \cdots d z_S d \theta \\
&=&  \int_{\mathcal{Z}_{1:S}} \left(  \int_{\Theta} \prod_{s=1}^S \tilde{p}(\theta|y_s, z_s) d \theta \right) \left( \prod_{s=1}^S \tilde{p}(z_s|y_s) \right) d z_1, \cdots d z_S, 
\end{eqnarray*}
where we have used Fubini's theorem. The quantity 
$\int_{\Theta} \prod_{s=1}^S \tilde{p}(\theta|y_s, z_s) d \theta$
is available in closed form and therefore we end up with the expression
\begin{eqnarray*}
&& I_{\text{sub}} \\ 
&=& \int_{\mathcal{Z}_{1:S}} \left(  \int_{\Theta} \prod_{s=1}^S \tilde{p}(\theta|y_s, z_s) d \theta \right) \left( \prod_{s=1}^S \tilde{p}(z_s|y_s) \right) d z_1, \cdots d z_S \\
&=& \mathbf{E}_{\tilde{p}(z_{1:S} | y_{1:S})} \left(  \int_{\Theta} \prod_{s=1}^S \tilde{p}(\theta|y_s, z_s) d \theta \right) .
\end{eqnarray*}
\end{proof}

\subsection{Proof of Proposition \ref{thm:propisubvar}}
\begin{proof}
We recall the definition of the effective joint distribution of the latent variables if sampling independently from each subposterior:
$$
    \tilde{p}(z_{1:S}|y_{1:S}) = \prod_{s=1}^S \tilde{p}(z_s|y_s).
$$
Now let's rewrite the full augmented posterior of $\theta$:
\begin{eqnarray*}
&& p(\theta|y_{1:S}, z_{1:S}) \\
&=& \frac{p(y_{1:S}, z_{1:S}|\theta)p(\theta)}{p(y_{1:S}, z_{1:S})}, \\
&=& \frac{\alpha^S}{p(z_{1:S}|y_{1:S}) p(y_{1:S})} \prod_{s=1}^S p(y_s, z_s | \theta) \tilde{p}(\theta),
\end{eqnarray*}
where we have used our result from \eqref{eq:prodalphaposterior} (conditioning on the latent variable). Now 
\begin{eqnarray*}
&=& \frac{\alpha^S}{p(z_{1:S}|y_{1:S}) p(y_{1:S})} \prod_{s=1}^S  \tilde{p}(\theta|y_s, z_s) \tilde{p}(z_s|y_s) \tilde{p}(y_s), \\
&=& \frac{\alpha^S \prod_{s=1}^S \tilde{p}(y_s)}{ p(y_{1:S})} \frac{\prod_{s=1}^S \tilde{p}(z_s|y_s)  }{p(z_{1:S}|y_{1:S})} \prod_{s=1}^S  \tilde{p}(\theta|y_s, z_s),
\end{eqnarray*}
where we have used $\tilde{p}(\theta|y_s, z_s) \tilde{p}(z_s|y_s) \tilde{p}(y_s) = p(y_s, z_s | \theta) \tilde{p}(\theta)$.
Now we use 
$p(y_{1:S}) = \left(\prod_{s=1}^S \tilde{p}(y_s) \right) \alpha^S I_{\text{sub}} $ yielding 
$$
\frac{\left(\prod_{s=1}^S \tilde{p}(y_s) \right)}{p(y_{1:S})} = \frac{\alpha^{-S}}{I_{\text{sub}} } .
$$
Inserting this in the above expression yields 
$$
p(\theta|y_{1:S}, z_{1:S}) =\\
\frac{1}{ I_{\text{sub}} } \frac{\prod_{s=1}^S \tilde{p}(z_s|y_s)  }{p(z_{1:S}|y_{1:S})} \prod_{s=1}^S  \tilde{p}(\theta|y_s, z_s) .
$$
Integrating both sides over $\theta$ yields
$$
1 =
\frac{1}{ I_{\text{sub}} } \frac{\prod_{s=1}^S \tilde{p}(z_s|y_s)  }{p(z_{1:S}|y_{1:S})} \int_\Theta \prod_{s=1}^S  \tilde{p}(\theta|y_s, z_s) d \theta, 
$$
and after rearranging we get
$$
\int_\Theta \prod_{s=1}^S  \tilde{p}(\theta|y_s, z_s) d \theta = I_{\text{sub}} 
\frac{p(z_{1:S}|y_{1:S})}{\prod_{s=1}^S \tilde{p}(z_s|y_s)  } , 
$$ 
We now use the estimator from \eqref{eq:isubestimator} and plug in the previous equation:
\begin{eqnarray*}
\hat{I}_{\text{sub}} &=& \frac{1}{N} \sum_{i=1}^N \int_\Theta \prod_{s=1}^S \tilde{p}(\theta|y_s, z_s^i) d \theta \\
&=& I_{\text{sub}}  \frac{1}{N} \sum_{i=1}^N 
\frac{p(z_{1:S}^i|y_{1:S})}{\prod_{s=1}^S \tilde{p}(z_s^i|y_s)  } \\
&=& I_{\text{sub}}  \frac{1}{N} \sum_{i=1}^N 
\frac{p(z_{1:S}^i|y_{1:S})}{\tilde{p}(z_{1:S}^i|y_{1:S})  }, 
\end{eqnarray*}
where we have used the definition $\tilde{p}(z_{1:S}^i|y_{1:S}) = \prod_{s=1}^S \tilde{p}(z_s^i|y_s)$ in the last line. 
Assuming independent sampling from the the latent variable subposterior we get an expression for the variance as 
$$
\Var \hat{I}_{\text{sub}} = \frac{I_{\text{sub}}^2}{N} 
\Var_{\tilde{p}(z_{1:S}|y_{1:S})} \left[ \frac{p(z_{1:S}|y_{1:S})}{\tilde{p}(z_{1:S}|y_{1:S})} \right].
$$
\end{proof}

\subsection{Proof of Proposition \ref{prop:isub_normal}}
Let $I_{\text{sub}}$ denote the true value of the subposterior integral, and let  $\widehat{I}_{\text{sub}}$ denote the proposed normal approximation:
\begin{align*}
    I_{\text{sub}} &= \int_{\Theta} \prod_{s=1}^{S}\tilde{p}(\theta | y_{s})  \ d\theta, \\
    \widehat{I}_{\text{sub}} &= \int_{\Theta} \prod_{s=1}^{S}\mathcal{N}(\theta | \mu_{s}, \Sigma_{s}) \ d\theta.
\end{align*}
The global normal approximation to the posterior $\mathcal{N}(\theta | \mu, \Sigma) \propto \prod_{s=1}^{S}\mathcal{N}(\theta|\mu_{s}, \Sigma_{s})$ has the the form
\begin{align}
    \mathcal{N}(\theta | \mu, \Sigma) &= \dfrac{\prod_{s=1}^{S}\mathcal{N}(\theta | \mu_{s}, \Sigma_{s})}{\widehat{I}_{\text{sub}}}. \label{eq:normal_appx}
\end{align}
A simple identity is
\begin{align}
    I_{\text{sub}}  &= \int_{\Theta} \prod_{s=1}^{S}\tilde{p}(\theta | y_{s})  \ d\theta, \\
    &= \int_{\Theta} \mathcal{N}(\theta | \mu, \Sigma) \dfrac{\prod_{s=1}^{S}\tilde{p}(\theta | y_{s})}{\mathcal{N}(\theta | \mu, \Sigma)}  \ d\theta. \label{eq:isub_identity}
\end{align}
Substituting  \eqref{eq:normal_appx} into \eqref{eq:isub_identity} yields
\begin{align*}
     I_{\text{sub}}  &= \widehat{I}_{\text{sub}}\int_{\Theta}\mathcal{N}(\theta | \mu, \Sigma) \dfrac{\prod_{s=1}^{S}\tilde{p}(\theta | y_{s})}{\prod_{s=1}^{S}\mathcal{N}(\theta | \mu_{s}, \Sigma_{s})}   \ d\theta \\
     &= \widehat{I}_{\text{sub}}\int_{\Theta}\mathcal{N}(\theta | \mu, \Sigma)\prod_{s=1}^{S} \dfrac{\tilde{p}(\theta | y_{s})}{\mathcal{N}(\theta | \mu_{s}, \Sigma_{s})}  \ d\theta.
\end{align*}
As $\tilde{p}({\theta} | y_{s})=0$ for $\theta \in \overline{\Theta}$, the domain of integration can be extended, and we can express the result in terms of an expectation. 
\begin{align*}
   I_{\text{sub}}  &= {\widehat{I}_{\text{sub}}}\int_{\Theta \cup \overline{\Theta}}\mathcal{N}(\theta | \mu, \Sigma)\prod_{s=1}^{S} \dfrac{\tilde{p}(\theta | y_{s})}{\mathcal{N}(\theta | \mu_{s}, \Sigma_{s})}  \ d\theta, \\
   &= {\widehat{I}_{\text{sub}}}{{\mathbb{E}_{\mathcal{N}({\theta}; {{\mu}}, {{\Sigma}})}\left[ \prod_{s=1}^{S}\dfrac{\widetilde{p}({\theta} \mid {y}_{s})}{\mathcal{N}({\theta}; {{\mu}}_{s}, {{\Sigma}}_{s})} \right]}}.
\end{align*}
Substituting the definitions of $I_{\text{sub}}$ and $\widehat{I}_{\text{sub}}$ gives the result presented in Proposition \ref{prop:isub_normal}.

\subsection{Proof of Proposition \ref{prop:log_error}}
We assume that the subposterior normal approximations satisfy the density ratio bounds
        \begin{align*}
            \max_{s=1, \ldots, S} \ \sup_{{\theta} \in {\Theta}} \dfrac{\widetilde{p}({\theta} \mid {y}_{s})}{\mathcal{N}({\theta}; {\mu}_{s}, {\Sigma}_{s})}  \le A, \qquad      \max_{s=1, \ldots, S} \ \sup_{{\theta} \in {\Theta}} \dfrac{\mathcal{N}({\theta}; {\mu}_{s}, {\Sigma}_{s})}{\widetilde{p}({\theta} \mid {y}_{s})} \le B,
        \end{align*}
Using the upper bound $A$ we have that
\begin{align}
    \mathbb{E}_{\mathcal{N}({\theta}; {{\mu}}, {{\Sigma}})}\left[ \prod_{s=1}^{S}\dfrac{\widetilde{p}({\theta} \mid {y}_{s})}{\mathcal{N}({\theta}; {{\mu}}_{s}, {{\Sigma}}_{s})} \right] &\le  \mathbb{E}_{\mathcal{N}({\theta}; {{\mu}}, {{\Sigma}})}\left[\prod_{s=1}^{S}A\right] = A^{S}. \label{eq:a_bound}
\end{align}
For $\theta \in \Theta$, it must hold that 
\begin{align*}
\prod_{s=1}^{S}\dfrac{\widetilde{p}({\theta} \mid {y}_{s})}{\mathcal{N}({\theta}; {{\mu}}_{s}, {{\Sigma}}_{s})} &\ge \prod_{s=1}^{S}\dfrac{1}{B}. %
\end{align*}
Using the law of total expectation we have that
\begin{align}
  \mathbb{E}_{\mathcal{N}({\theta}; {{\mu}}, {{\Sigma}})}\left[ \prod_{s=1}^{S}\dfrac{\widetilde{p}({\theta} \mid {y}_{s})}{\mathcal{N}({\theta}; {{\mu}}_{s}, {{\Sigma}}_{s})} \right] &=   \mathbb{E}_{\mathcal{N}({\theta}; {{\mu}}, {{\Sigma}})}\left[  \prod_{s=1}^{S}\dfrac{\widetilde{p}({\theta} \mid {y}_{s})}{\mathcal{N}({\theta}; {{\mu}}_{s}, {{\Sigma}}_{s})} \ \middle\vert \   \theta \in \Theta \right]\mathbb{E}_{\mathcal{N}({\theta}; {{\mu}}, {{\Sigma}})}[\mathbbm{1}(\theta \in \Theta)] \nonumber \\
  &\ge  \mathbb{E}_{\mathcal{N}({\theta}; {{\mu}}, {{\Sigma}})}\left[ 1/B^{S} \mid   \theta \in \Theta \right]\mathbb{E}_{\mathcal{N}({\theta}; {{\mu}}, {{\Sigma}})}[\mathbbm{1}(\theta \in \Theta)] \nonumber \\
  &= B^{-S}\mathbb{E}_{\mathcal{N}({\theta}; {{\mu}}, {{\Sigma}})}[\mathbbm{1}(\theta \in \Theta)]. \label{eq:b_bound}
\end{align}
Therefore
\begin{align*}
  -S \log B  + \log \mathbb{E}_{\mathcal{N}({\theta}; {{\mu}}, {{\Sigma}})}[\mathbbm{1}(\theta \in \Theta)] \le \log \dfrac{\int_{\Theta} \prod_{s=1}^{S}\tilde{p}(\theta | y_{s})  \ d\theta}{\int_{\Theta} \prod_{s=1}^{S}\mathcal{N}(\theta | \mu_{s}, \Sigma_{s}) \ d\theta} \le S \log A.
\end{align*}
Multiplying the numerator and denominator by $\alpha^S\prod_{s=1}^{S}\tilde{p}(y_{s})$ gives
\begin{align*}
  -S \log B + \log \mathbb{E}_{\mathcal{N}({\theta}; {{\mu}}, {{\Sigma}})}[\mathbbm{1}(\theta \in \Theta)] \le \log \dfrac{\alpha^S\prod_{s=1}^{S}\tilde{p}(y_{s})\int_{\Theta} \prod_{s=1}^{S}\tilde{p}(\theta | y_{s})  \ d\theta}{\alpha^S\prod_{s=1}^{S}\tilde{p}(y_{s})\int_{\Theta} \prod_{s=1}^{S}\mathcal{N}(\theta | \mu_{s}, \Sigma_{s}) \ d\theta} \le S \log A.
\end{align*}
Then using the definition of $\widehat{p}(y)$ and Proposition \ref{thm:prop1} we have the final result
\begin{align*}
  -S \log B + \mathbb{E}_{\mathcal{N}({\theta}; {{\mu}}, {{\Sigma}})}[\mathbbm{1}(\theta \in \Theta)] \le \log \dfrac{p(y)}{\widehat{p}(y)} \le S \log A.
\end{align*}

\subsection{Proof of Corollary \ref{cor:growth}}
From Proposition \ref{prop:isub_normal}, the correction factor $p(y)/\hat{p}(y)$ can be written as
\begin{align}
  \mathbb{E}_{\mathcal{N}({\theta}; {{\mu}}, {{\Sigma}})}\left[ \prod_{s=1}^{S}\dfrac{\widetilde{p}({\theta} \mid {y}_{s})}{\mathcal{N}({\theta}; {{\mu}}_{s}, {{\Sigma}}_{s})}\right] &= 
    \mathbb{E}_{\mathcal{N}({\theta}; {{\mu}}, {{\Sigma}})}\left[ \prod_{j=1}^{S} \exp\left( \dfrac{1}{S} \sum_{s=1}^{S}\log \dfrac{\tilde{p}(\theta | y_{s})}{\mathcal{N}(\theta | \mu_{s}, \Sigma_{s})}\right)\right].  \label{eq:exp_identity}
\end{align}
We  assume  that
\begin{align*}
    \sup_{{\theta} \in {\Theta}}  \left\lvert \dfrac{1}{S}\sum_{s=1}^{S}\log \dfrac{\widetilde{p}(\theta | y_{s})}{\mathcal{N}(\theta | \mu_{s}, \Sigma_{s})} \right\rvert &= \OO_{p}(1)
\end{align*}
where $y_{s}, \mu_{s}$ and $\Sigma_{s}$ are treated as random. This implies that for any $\epsilon > 0$ there exists a finite $M > 0$ and $S' > 0$ such that
\begin{align*}
    &\Pr\left( \exp(-M)  \le \sup_{{\theta} \in {\Theta}}  \exp\left\lbrace\dfrac{1}{S} \sum_{s=1}^{S} \log \dfrac{\widetilde{p}(\theta | y_{s})}{\mathcal{N}(\theta | \mu_{s}, \Sigma_{s})} \right\rbrace \le \exp(M)\right) \ge  1-\epsilon,  \quad \forall   S > S'
\end{align*}
Therefore, it holds with probability at least $1-\epsilon$ for $S > S'$,
\begin{align*}
\mathbb{E}_{\mathcal{N}({\theta}; {{\mu}}, {{\Sigma}})}\left[ \prod_{j=1}^{S} \exp(-M)\right] \le \mathbb{E}_{\mathcal{N}({\theta}; {{\mu}}, {{\Sigma}})}\left[ \prod_{j=1}^{S} \exp\left( \dfrac{1}{S} \sum_{s=1}^{S}\log \dfrac{\tilde{p}(\theta | y_{s})}{\mathcal{N}(\theta | \mu_{s}, \Sigma_{s})}\right)\right]  &\le \mathbb{E}_{\mathcal{N}({\theta}; {{\mu}}, {{\Sigma}})}\left[ \prod_{j=1}^{S} \exp(M)\right], \\
\exp(-SM) \le \mathbb{E}_{\mathcal{N}({\theta}; {{\mu}}, {{\Sigma}})}\left[ \prod_{j=1}^{S} \exp\left( \dfrac{1}{S} \sum_{s=1}^{S}\log \dfrac{\tilde{p}(\theta | y_{s})}{\mathcal{N}(\theta | \mu_{s}, \Sigma_{s})}\right)\right]  &\le \exp(SM).
\end{align*}
As such, for any $\epsilon > 0$ there exists a finite $M > 0$ and $S' > 0$ such that with probability at least $1-\epsilon$,
\begin{align*}
 -SM \le \log  \mathbb{E}_{\mathcal{N}({\theta}; {{\mu}}, {{\Sigma}})}\left[ \prod_{s=1}^{S}\dfrac{\widetilde{p}({\theta} \mid {y}_{s})}{\mathcal{N}({\theta}; {{\mu}}_{s}, {{\Sigma}}_{s})}\right] \le SM, \qquad \forall S > S'
\end{align*}
where we have made use of \eqref{eq:exp_identity}. Therefore, we conclude
\begin{align*}
 \log \dfrac{p(y)}{\widehat{p}(y)} &=  \log  \mathbb{E}_{\mathcal{N}({\theta}; {{\mu}}, {{\Sigma}})}\left[ \prod_{s=1}^{S}\dfrac{\widetilde{p}({\theta} \mid {y}_{s})}{\mathcal{N}({\theta}; {{\mu}}_{s}, {{\Sigma}}_{s})}\right] = \OO_{p}(S).
 \end{align*}

\end{document}